\documentclass[a4]{svjour3}

\usepackage{amsfonts}
\usepackage{amsmath}
\usepackage{amssymb}
\usepackage{mathtools}
\usepackage{graphicx}
\usepackage{bm}
\usepackage[authoryear]{natbib}
\usepackage{xcolor}
\usepackage{graphicx}
\usepackage{tikz}
\usepackage{lineno}

\smartqed

\journalname{ }

\title{The transition distribution of a sample from a Wright-Fisher diffusion with general small mutation rates}

\titlerunning{Wright-Fisher diffusion with small general mutation rates}

\author{Conrad J.\ Burden \and Robert C. Griffiths}
\institute{Conrad J.\ Burden, Mathematical Sciences Institute, Australian National University, Canberra, Australia
\at \email{conrad.burden@anu.edu.au}
\and
Robert C.\ Griffiths, School of Mathematics, Monash University, Australia; Corresponding author
\at
\email{Bob.Griffiths@Monash.edu}
}

\date{Received: date / Accepted: date}

\begin{document}
\maketitle
\begin{abstract}
The transition distribution of a sample taken from a Wright-Fisher diffusion with general small mutation rates is found using a coalescent approach. The approximation  is equivalent to having at most one mutation in the coalescent tree of the sample up to the most recent common ancestor with additional mutations occurring on the lineage from the most recent common ancestor to the time origin if complete coalescence occurs before the origin. The sampling distribution leads to an approximation for the transition density in the diffusion with small mutation rates.
This new solution has interest because the transition density in a Wright-Fisher diffusion with general mutation rates is not known. 

\end{abstract}

%
%
\section{Introduction}
\label{sec:Introduction}

\citet{BT2016,BT2017} find an approximation for the stationary distribution in a $d$-allele neutral Wright-Fisher diffusion with general small mutation rates.
The model has a connection to boundary processes in which only transitions involving non-segregating and bi-allelic states are considered.
\citet{SH2017},\citet{VB2015,DSK2015} study this boundary mutation property. 

The Kingman coalescent process is dual to the Wright-Fisher diffusion in describing the ancestral history of a sample of $n$ individuals in a population back in time. 
In \citet{BG2019} a coalescent approach with mutations in the tree is used to find an approximate sampling formula for small mutation rates which then leads to an approximation for the stationary distribution in the diffusion.
They show that finding approximate formulae for small rates is equivalent to considering at most one mutation in a coalescent tree. 

In the current paper we find the sampling distribution in the Wright-Fisher diffusion with general mutation rates at time $t\geq 0$ with small rates by using a coalescent approach. 
The probability distribution of a gene tree at time $t$  is found in a similar way. A gene tree is equivalent to a sample of sequences in the infinitely-many-sites model.

 The details of the sampling distribution in the general mutation model are more complex than in a stationary model because the coalescent tree back in time is constrained by the origin. The sampling distribution leads to an approximation for the transition density in the general mutation model which is of interest as a solution to the diffusion process, in population genetics, and more widely as an approximate differential equation solution. 
\citet{V2014}, Theorem~3, considers a spectral expansion of the transition function in a two-allele model when mutation rates are small.
The eigenfunctions are known in a two-allele model, but are unknown if there are more than two alleles, so this approach is not used here.

\citet{BG2018b} also find an approximation to the stationary density in a two island, two allele model when mutation and migration rates are small by using a flux argument in the Wright-Fisher diffusion process as well as a coalescent argument.

%
%
\section{Preliminaries}
\label{sec:Preliminaries}

A $d$-dimensional Wright-Fisher diffusion process $\{\bm{X}(t)\}_{t\geq 0}$ with general mutation rates has a generator
\begin{equation}
\mathcal{L}=\frac{1}{2}\sum_{i,j=1}^dx_i(\delta_{ij}-x_j)\frac{\partial^2}{\partial x_i\partial x_j}+\sum_{i,j=1}^d\gamma_{ji}x_j\frac{\partial}{\partial x_i},
\label{gen:0}
\end{equation}
where $\gamma=(\gamma_{ij}) = \frac{\theta}{2}(P-I)$,  where $\theta/2$ is the overall mutation rate away from each type and $P$ a $d\times d$ transition probability matrix for type changes which has non-negative entries and row sums $1$.
There are different possible pairs $(\theta,P)$ corresponding to a fixed rate matrix $\gamma$ with non-negative off-diagonal entries and row sums $0$.
  Note that for a choice $(\theta,P)$ with $\gamma$ fixed
$P=\frac{2}{\theta}\gamma + I$, so  $\theta$ has to be chosen so the off-diagonal entries of $P$ are non-negative. That is, choose $\theta$ large enough so that $\frac{2}{\theta}\max_{1 \leq i \leq d}(-\gamma_{ii}) \leq 1$. The minimal value of $\theta$ in choosing a pair $(\theta,P)$ for a given $\gamma$ is $\theta = 2 \max_{1\leq i \leq d}(-\gamma_{ii})$. This may define an appropriate pair $(\theta,P)$ for us to use in this paper because we are considering small mutation rates to order $\theta$. It is possible to use larger values of $\theta$, for example 
$\theta=\sum_{i\ne j}\gamma_{ij} = \sum_{i=1}^d(-\gamma_{ii})$ which is like a matrix norm. There may be other biological reasons to consider a particular $(\theta,P)$ pair and regard $\gamma$ as being determined by that pair.
Explicit analytic solutions for the stationary distribution, \citet{W1969}, and the transition density in the diffusion, \citet{G1979,EG1993,GS2010}, are only known for parent independent mutation when $P$ has identical rows, in which case the stationary density is a  Dirichlet distribution.  

The coalescent duality with the Wright-Fisher diffusion is well known in population genetics. Briefly the argument is the following. Denote the multinomial distribution by 
\begin{equation}
m(\bm{n},n,\bm{p}) = {n \choose \bm{n}}\prod_{j=1}^dp_j^{n_j},	\label{multinomial}
\end{equation}
where $\bm{n} = (n_1, \ldots, n_d)$, $\sum_{j = 1}^d n_j = n$.  Let 
\[
p(\bm{n};t,\bm{x}) = \mathbb{E}\big [m(\bm{n};n,\bm{X}(t))\mid \bm{X}(0)=\bm{x}\big]
\]
 be the probability distribution of a configuration $\bm{n}$ in a sample of $n$ individuals taken at time $t$ when the initial frequencies in the diffusion are $\bm{x}$. 
Note that
\begin{eqnarray}
&&\mathcal{L}x_1^{n_1}\cdots x_d^{n_d} =
\nonumber \\
&& \frac{1}{2}\sum_{i=1}^dn_i(n_i-1)x_1^{n_1}\cdots x_i^{n_i-1}\cdots x_d^{n_d}
-\frac{1}{2}\sum_{i,j=1}^dn_i(n_j-\delta_{ij})x_1^{n_1}\cdots x_d^{n_d}
\nonumber\\
\nonumber\\
&&+\frac{\theta}{2}\sum_{i,j=1}^dP_{ji}n_i x_1^{n_1}\cdots x_i^{n_i-1+\delta_{ij}}\cdots x_j^{n_j+1-\delta_{ij}}\cdots  x_d^{n_d}  
-\frac{\theta}{2}\sum_{i=1}^dn_ix_1^{n_1}\cdots x_d^{n_d}.
\nonumber \\
\label{stat:1}
\end{eqnarray}
Simplifying Eq.~(\ref{stat:1}) and using basic theory 
\begin{eqnarray}
\frac{d}{dt}p(\bm{n};t,\bm{x}) &=& -\frac{1}{2}n(n-1+\theta)
p(\bm{n};t,\bm{x}) + \frac{1}{2}n(n-1)\sum_{i=1}^d\frac{n_i-1}{n-1}p(\bm{n}-\bm{e}_i;t,\bm{x})
\nonumber \\
&&
+ \frac{1}{2}n\theta\sum_{i,j=1}^d\frac{n_j+1-\delta_{ij}}{n}P_{ji}p(\bm{n}-\bm{e}_i + \bm{e}_j;t,\bm{x}).
\label{stat:2}
\end{eqnarray}
Arguing probabilistically Eqs.~(\ref{stat:2}) for $p(\bm{n};t,\bm{x})$ are the same as those for the sampling distribution in the coalescent with general mutations governed by $\theta$ and $P$ on the edges. The overall coalescence rate when $n$ individuals is $\frac{1}{2}n(n-1)$ and the overall mutation rate is $\frac{1}{2}n\theta$. If the first event back in time was a  coalescence event  then from the $n-1$ individuals after coalescence the configuration must be $\bm{n}-\bm{e}_i$ and a type $i$ individual chosen to give birth with probability $(n_i-1)/(n-1)$ to obtain a configuration of $\bm{n}$ for $i \in [d]$. If the first event back in time was a mutation, then from the $n$ individuals after the event to obtain a configuration of $\bm{n}$ a type $j$ individual mutates to a type $i$ individual from a configuration $\bm{n}-\bm{e}_i + \bm{e}_j$ with probability 
$(n_j+1-\delta_{ij})/n$ for $i,j \in [d]$.

\citet{BG2019} find that for small mutation rates, $\theta \rightarrow 0$, the probability of a configuration of $\bm{n}$ in a sample of $n$ individuals taken in the stationary 
distribution of the diffusion with generator in Eq.~(\ref{gen:0}),  $p(\bm{n};\theta) = \lim_{t \rightarrow \infty} p(\bm{n};t,\bm{x})$, is
\begin{eqnarray}
p(n\bm{e}_a;\theta) &=& \pi_a\Big (1 - \theta(1-P_{aa})\sum_{l=1}^{n-1}\frac{1}{l}\Big ) + \mathcal{O}(\theta^2)
\nonumber \\
p(n_a\bm{e}_a+n_b\bm{e}_b;\theta) &=& \theta\Big (\pi_aP_{ab}\frac{1}{n_b} + \pi_bP_{ba}\frac{1}{n_a}\Big ) + \mathcal{O}(\theta^2)
\nonumber \\
p(\bm{n};\theta) &=& \mathcal{O}(\theta^2)\text{~if~}\bm{n}\text{~has~}>2\text{~non-zero~entries},
\label{thm:c}
\end{eqnarray}
where $\bm{\pi}$ is the stationary distribution of $P$.  

The main result of this paper, Theorem~\ref{theorem:1} in Section~\ref{sec:SamplingArbitraryN}, extends the sampling formula to finite times $0 \le t < \infty$.  
Suppose a sample of genes taken at time $t$ is composed of families of genes from $l\geq 1$ founder lineages at time 0.  Fig. \ref{fig:lineages} illustrates 
two possibilities contributing to the probability $p(\bm{n};t,\bm{x})$, which need to be considered.  If there are $l > 1$ founder lineages, then to $\mathcal{O}(\theta)$ there 
is at most one mutation event in the coalescent tree, and hence at most one mutant family arising from a mutation.  However, if there is $l = 1$ founder lineage, 
then it turns out that to $\mathcal{O}(\theta)$ there is at most one mutation event in that part of the coalescent tree between the most recent common ancestor 
of the sample (MRCA) and time $t$, but any number of coalescent events along the single founder edge from time zero to the MRCA.  The reason for this 
is related to the asymptotic transition to the stationary distribution Eq.~(\ref{thm:c}) as $t \rightarrow \infty$, as will become 
apparent in the next section.  
 The proof of Theorem~\ref{theorem:1} uses this description and works through the combinatorial details of the type and size of families in a sample. Unlike in the stationary distribution there can be more than two types present when there are  greater than two families of founder lineages.
 \bigskip
 
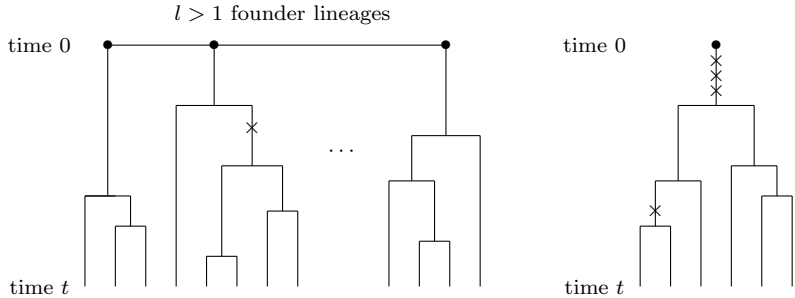
\begin{figure}
\begin{center}
\caption{Families from founder lineages}\label{fig:lineages}
\bigskip

\begin{tikzpicture}[xscale=0.4,yscale = 0.4]
\draw (0,0) -- (0,3);
\draw (0.75,3) -- (0.75,8);
\draw (1,0) -- (1,2);
\draw (1.5,2) -- (1.5,3);
\draw (2,0) -- (2,2);
\draw (3,0) -- (3,6);
\draw (4,0) -- (4,1);
\draw (5,0) -- (5,1);
\draw (4.5,1) -- (4.5,4);
\draw  (6,0) -- (6,2.5);
\draw  (7,0) -- (7,2.5);
\draw  (6.5,2.5) -- (6.5,4);
\draw  (5.5,4) -- (5.5,6);
\draw  (4.25,6) -- (4.25,8);
\draw  (10,0) -- (10,3.5);
\draw  (11,0) -- (11,1.5);
\draw  (12,0) -- (12,1.5);
\draw  (11.5,1.5) -- (11.5,3.5);
\draw  (10.75,3.5) -- (10.75,5);
\draw   (11.875,5) -- (11.875,8);
\draw  (13,0) -- (13,5);
\draw (0,3) -- (1,3);
\draw (1,2) -- (2,2);
\draw (0,3) -- (1.5,3);
\draw (4,1) -- (5,1);
\draw (6,2.5) -- (7,2.5);
\draw  (4.5,4) -- (6.5,4);
\draw  (3,6) -- (5.5,6);
\draw  (10,3.5) -- (11.5,3.5);
\draw  (11,1.5) -- (12,1.5);
\draw (10.75,5) -- (13,5);
\draw (0.75,8) -- (11.875,8);
\draw -- (8.5,4.5) node {$\cdots$};
\draw -- (5.5,5.25) node {$\times$};
\draw -- (6.5,9) node {$l>1$ founder lineages};
\draw -- (0.75,8) node {$\bullet$};
\draw -- (4.25,8) node {$\bullet$};
\draw -- (11.875,8) node {$\bullet$};
\draw -- (-1.5,8) node {time $0$};
\draw -- (-1.5,0) node {time $t$};
\end{tikzpicture}
\hfil
\begin{tikzpicture}[xscale=0.4,yscale = 0.4]
\draw (0,0) -- (0,2);
\draw (1,0) -- (1,2);
\draw (0.5,2)-- (0.5,3.5);
\draw (2,0) -- (2,3.5);
\draw (1.25,3.5) -- (1.25,6);
\draw (2.5,6) -- (2.5,8);
\draw (3,0) -- (3,4);
\draw (4,0) -- (4,3);
\draw (5,0) -- (5,3);
\draw (4.5,3) -- (4.5,4);
\draw (3.75,4) -- (3.75,6);
\draw (0,2) -- (1,2);
\draw (0.5,3.5) -- (2,3.5);
\draw (1.25,6) -- (3.75,6);
\draw (3,4) -- (4.5,4);
\draw (4,3) -- (5,3);
\draw -- (2.5,8) node {$\bullet$};
\draw -- (2.5,6.5) node {$\times$};
\draw -- (2.5,7.5) node {$\times$};
\draw -- (2.5,7) node {$\times$};
\draw -- (0.5,2.5) node {$\times$};
\draw -- (-1.5,8) node {time $0$};
\draw -- (-1.5,0) node {time $t$};
\end{tikzpicture}
\end{center}
\bigskip

The first illustration shows how genes in a sample are partitioned by families from $l>1$ founder lineages. To $\mathcal{O}(\theta)$ in probability there can be at most one mutation along sample lineages, indicated by $\times$.

 The second illustration shows the case when there is just one founder lineage of the genes in the sample. Similarly to $\mathcal{O}(\theta)$ in probability there can be at most one mutation in the sample lineages between the MRCA and time $t$, but any number of mutations along the single lineage from the origin to the MRCA.

\end{figure}
\bigskip
 
 The number of edges in a coalescent tree of $n$ individuals is a pure death process beginning at $n$ with death rates
 $\mu_k = {k\choose 2}$, $k=n,\ldots, 2$.  The times between events are exponential random variables $T_k$ with rates $\mu_k$, 
 $k =n,\ldots ,2$. The number of edges in a coalescent tree at time $t$ back from the present time, $A_n(t)$, has a probability distribution
\begin{equation}
q_{nl}(t) := P(A_n(t)=l) = \sum_{j=l}^n\rho_j(t)(-1)^{j-l}
\frac{
(2j-1)(j+l-2)!
}
{
l!(l-1)!(j-l)!
}
\frac{n_{[j]}}{n_{(j)}},
\label{lod:0}
\end{equation}
where the notation is that 
\begin{equation*}
\begin{split}
n_{(j)} = n(n+1)\cdots (n+j-1), \\
n_{[j]} = n(n-1)\cdots (n-j+1),
\end{split}
\end{equation*}
and $\rho_j(t) = \exp \{-{j\choose 2}t\}$, \citep{T1984,G1980}. Denote $q_{nl}(t) = \mathbb{P}(A_n(t)=l)$ and the probability density of $T_n+\cdots +T_l$ as $f_{nl}(t)$.  Then 
\begin{eqnarray*}
f_{nl}(t) \,dt & = & \mathbb{P}\left( T_n + \cdots + T_l \in (t, t + dt) \right) \\
	& = & \mathbb{P}\left(A_n(t) = l\right) \times \text{Prob.\ of coalescence from $l$ to $l - 1$ in $(t, t + dt)$} \\
	& = & q_{nl}(t) \times {l \choose 2} dt, 
\end{eqnarray*} 
from which follows the identity
\begin{equation}
q_{nl}(t) = {l\choose 2}^{-1}f_{nl}(t).
\label{identity:100}
\end{equation}
As $n \rightarrow \infty$, $q_{\infty l}(t)$ and $f_{\infty l}(t)$ are proper distributions where $n_{[j]}/n_{(j)}$ is replaced by 1 in Eq.~(\ref{lod:0}).

Combinatorics of the edge configuration in a coalescent tree are related to a Polya urn model. The probability $r_{\bm{n}\mid \bm{l}}$ of a configuration of edges $\bm{l}$ 
of $d$ types giving rise to a configuration $\bm{n}$ ($n > l$) is equivalent to considering a Polya urn beginning with $\bm{l} > \bm{0}$ coloured balls in which $\bm{n}-\bm{l}$ 
draws of the $d$ balls are made. If a ball of colour $j\in [d]$ is chosen on a draw from the urn, then it is replaced together with an additional ball of colour $j$.
From classical theory, for which see \citet{GT2003},
\begin{eqnarray}
r_{\bm{n}\mid\bm{l}} 
&=& {n-l\choose \bm{n}-\bm{l}}\cdot 
\frac{
\prod_{a\in [d]} (l_a)_{(n_a-l_a)}
}
{l_{(n-l)}}
\nonumber \\
&= &{n-l\choose \bm{n}-\bm{l}}\cdot\frac{(l-1)!}{\prod_{a=1}^d (l_a-1)!}
\cdot \frac{\prod_{a=1}^d(n_a-1)!}{(n-1)!}
\nonumber \\
&=& {\cal DM}_{\bm{l},n-l}(\bm{n}-\bm{l}).	
\label{reqn:0}
\end{eqnarray}
Here ${\cal DM}$ is the Dirichlet-multinomial distribution, 
\begin{eqnarray*}
{\cal DM}_{\bm{\theta},n}(\bm{n}) &=& \int_{\Delta_d}{n \choose \bm{n}}x_1^{n_1}\cdots x_d^{n_d}{\cal D}_{\bm{\theta}}(d\bm{x})
\nonumber \\
&=& 
{n \choose \bm{n}} 
\frac
{
{\theta_1}_{(n_1)}\cdots {\theta_d}_{(n_d)}
}
{
\theta_{(n)}
},
\end{eqnarray*}
defined in terms of the Dirichlet distribution,  
\[
{\cal D}_{\bm{\theta}}(\bm{x}) = \frac{\Gamma (\theta)}{\Gamma(\theta_1)\cdots\Gamma (\theta_d)}x_1^{\theta_1 -1}\cdots x_d^{\theta_d-1},
\]
where 
$
\bm{x} \in \Delta_d = \{ 0 < x_i < 1, i \in [d], |\bm{x}|=1\}
$.
That is, ${\cal DM}_{\bm{\theta},n}(\bm{n})$ is probability of drawing a configuration $\bm{n}$ in a sample of $n$ individuals apportioned with 
random probabilities $\bm{X} \sim \text{Dirichlet}(\bm{\theta})$. 
 
It is useful to extend the definition of $r_{\bm{n}\mid \bm{l}}$ to allow some entries in $\bm{l}$ to be zero. Then $n_i=0$ if $l_i=0$ and
\begin{equation}
r_{\bm{n}\mid\bm{l}} = {n-l\choose\bm{n}-\bm{l}}\cdot
\frac{(l-1)!}{\prod_{l_a>0}(l_a-1)!}\cdot
\frac{\prod_{l_a>0}(n_a-1)!}{(n-1)!}.
\label{extend:0}
\end{equation}
In Theorem~\ref{theorem:1} the notation $\bm{n}\backslash ij$, $\bm{l}\backslash ij$ means the vectors with their $i$ and $j$th entries deleted. $r_{\bm{n}\backslash ij\mid \bm{l}\backslash i j}$ is equal to Eq.~(\ref{extend:0}) with $l_i=0$, $l_j=0$ and $n\backslash ij \equiv n-n_i-n_j$. A calculation shows that $r_{\bm{n}\backslash ij\mid \bm{l}\backslash i j}$ is equal to the conditional probability
$r_{\bm{n}\mid\bm{l}}/\mathbb{P}(n_i,n_j\mid l_i,l_j)$
where the marginal distribution of $n_i,n_j$ is
\begin{eqnarray}
\mathbb{P}(n_i,n_j\mid \bm{l}) &=&
\mathbb{P}(n_i,n_j\mid l_i,l_j) 
\nonumber \\
&=& \sum_{n_a, a \in [d]\backslash \{i, j\}}
r_{\bm{n}\mid \bm{l}}
\nonumber \\
&=& {\cal DM}_{(l-l_i-l_j,l_i,l_j),n-l}(n\backslash ij - l\backslash ij,n_i-l_i,n_j-l_j).
\label{configuration:xx}
\end{eqnarray}
An easy way to see this is to consider a Polya urn in which all colours other than $i$ and $j$ are considered identical.  
The probability that a particular edge, while $k$ edges in a coalescent tree, subtends $c$ edges in a sample of $n$ is a special case of Eq.~(\ref{extend:0})
\begin{eqnarray}
p_{nk}(c) &=& r_{(c,n-c)\mid(1,k-1)}
\nonumber \\
&= &\frac{{n-c-1\choose k-2}}{{n-1\choose k-1}},\>1 \leq c \leq n-k+1.
\label{pnkc:0}
\end{eqnarray} 
See also \citet{GT1988}.

%
%

\section{Sampling Distribution for Sampling Size $n = 2$}
\label{sec:SamplingNEquals2}

 To help understand Theorem~\ref{theorem:1}, we first show how to calculate the probability of possible configurations when $n=2$. 
Looking back in time from $t$ to time zero there are $l=2$ founder lineages with probability $e^{-t}$ or $l=1$ founder lineages when the two initial lineages 
coalesce before $t$ with probability $1-e^{-t}$.  

Consider first the contribution to the probability of a sample configuration $n_a=2$, at time $t$, when there are $l = 2$ founder lineages.  If we assume that $t = o(1/\theta)$, indicated by $t << 1/\theta$, then the probability that $l=2$ and that greater or equal to two mutations occur is $e^{-t}\big (1 - e^{-\theta t} - \theta t e^{-\theta t}\big) = \mathcal{O}(\theta^2)$. 
For larger times, $t = \mathcal{O}(1/\theta)$, the probability
$e^{-t}$ is $o(\theta)$, since $\lim_{\theta \to 0} e^{-1/\theta}/\theta = 0$, and the $l = 2$ contribution can safely be ignored to $\mathcal{O}(\theta)$.
%
%
%
%
%
The probability that $l=2$ and $n_a=2$ is therefore
\begin{eqnarray}
&&x_a^2e^{-\theta t}e^{-t} + \big (2x_a^2P_{aa} + \sum_{j\ne a}2x_ax_jP_{ja}\big )	\frac{\theta}{2}te^{-\theta t}e^{-t} + o(\theta)
\nonumber \\
&&~~= x_a^2(1-\theta t)e^{-t} + \big (\sum_{j\in [d]}2x_ax_jP_{ja}\big )\frac{\theta}{2}te^{-t} + o(\theta).
\label{twol:2}
\end{eqnarray}
In (\ref{twol:2}) there are two founder lineages at time $t$ with probability $e^{-t}$; the first term corresponds to the probability that both founder lineages are of type $a$ and no mutation occurs in $(0,t)$; the second term corresponds to either both lineages being of type $a$ and a mutation occurring from $a$ to $a$, not changing the type of the lineage, or one founder lineage is of type $a$ and the other is of type $j\ne a$ and a mutation occurs changing the lineage from type $j$ to type $a$. The probability of one mutation in $(0,t)$ on a single edge is $\frac{\theta}{2}te^{-\theta t} + o(\theta)$.

If $l=1$, the probability $1-e^{-t}$ that the two lineages coalesce is not small for large $t$, and care must be taken to include multiple mutations in the coalescent tree.  
Suppose coalescence occurs at time $w<t$ back. The two edges of the coalescent tree from the coalescent event at $t-w$ to the current time $t$ contribute a 
factor $e^{-w}$, which is $\mathcal{O}(1)$ for $w << 1/\theta$ and $o(\theta)$ for $w = \mathcal{O}(1/\theta)$.  Thus it is sufficient to consider at most one mutation 
in this part of the tree.  On the other hand, there is no restriction on the number of mutations on the single edge from time 0 to the coalescent event if the 
calculation is to be accurate to $\mathcal{O}(\theta)$ for large $t-w$, in particular, for $(t-w)\theta = \mathcal{O}(1)$.  
Writing the mutation rate matrix as $\gamma = \frac{\theta}{2}(P - I)$, 
the probability that $l=1$ and $n_a=2$ for all $t > 0$ is therefore
\begin{eqnarray}
\lefteqn{ \sum_{j \in [d]} x_j \int_0^t \left(e^{(t - w)\gamma} \right)_{ja} e^{-\theta w} (1 + P_{aa}\theta w) e^{-w} dw + o(\theta) } \nonumber \\
& = & \sum_{i \in [d]} \left(\bm{x} e^{t\gamma} \right)_i \int_0^t \left(e^{-w \gamma} \right)_{ia} e^{-\theta w} (1 + P_{aa}\theta w) e^{-w} dw + o(\theta) \nonumber \\
& = & \sum_{i \in [d]} \left(\bm{x} e^{t\gamma} \right)_i \int_0^t \left(\delta_{ia} - \tfrac{1}{2} \theta w (P_{ia} - \delta_{ia}) \right) (1 - \theta w)(1 + P_{aa} \theta w) e^{-w} dw + o(\theta) 
                                        \nonumber \\
& = &  \left(\bm{x} e^{t\gamma} \right)_a \int_0^t  e^{-w} dw + \nonumber \\
&&      \qquad  \theta \sum_{i \in [d]} \left(\bm{x} e^{t\gamma} \right)_i \left\{ \delta_{ia} (P_{aa} - 1) - \tfrac{1}{2}(P_{ia} - \delta_{ia}) \right\} \int_0^t  we^{-w} dw  + o(\theta)     \nonumber \\
& = &  \left(\bm{x} e^{t\gamma} \right)_a (1 - e^{-t}) + \nonumber \\
&&     \qquad                \theta  \left\{\left(\bm{x} e^{t\gamma} \right)_a (P_{aa} - 1) - \tfrac{1}{2} \left(\bm{x} e^{t\gamma} (P - I) \right)_a \right\} \left( 1 - (1 + t)e^{-t} \right) + o(\theta).  \nonumber \\ 
\label{twol:1}
\end{eqnarray}
Probabilistically equation (\ref{twol:1}) is easiest to read from the first line, rather than from the simplified last line. In the first line the single founder is chosen to be type $j$ with probability $x_j$; the MRCA of the sample lineages is type $a$ from mutations along the single lineage; there could be no mutation along the two sample lineages, or one mutation from type $a$ to type $a$, not changing the type of the lineage. The integral accounts for the probability of coalescence of the two sample lineages at $w\in (0,t)$.

The total probability of a configuration $n_a=2$ is the sum of Eqs.~(\ref{twol:2}) and (\ref{twol:1}).  The stationary distribution is recovered by 
noting that $\bm{x} e^{t\gamma} \rightarrow \bm{\pi}$ as $t \rightarrow \infty$, giving 
$$
p(2 \bm{e}_a; \theta) = \pi_a (1 - \theta (1 - P_{aa}) ) + o(\theta).  
$$
If $t << 1/\theta$, then $e^{t\gamma} = I  + \tfrac{1}{2} \theta t (P - I) + o(\theta)$, and Eqs.~(\ref{twol:2}) and  (\ref{twol:1}) give  
\begin{eqnarray} 
\lefteqn{p(2 \bm{e}_a; t, \bm{x}) = x_a^2(1-\theta t)e^{-t} + \left(\sum_{j\in [d]}2x_ax_jP_{ja}\right)\frac{\theta}{2}te^{-t} +\ x_a(1 - e^{-t})\ +} \nonumber \\
	&&  \theta\left\{\tfrac{1}{2}(t - 1 + e^{-t}) \left(\sum_{j\in [d]} x_iP_{ia} - x_a \right) + (1 - (1 + t)e^{-t}) x_a (P_{aa} - 1) \right\} + o(\theta).  \nonumber \\
\label{twol:smalltheta}
\end{eqnarray}

The other possible configuration of two types is $n_a=1$, $n_b=1$, $a\ne b$.  If there are $l=2$ founder lineages, then by an argument similar to that for the $n_a=2$ case  
it is sufficient to consider at most one mutation in the tree.  Then either the two founder lineages are of types $a$ and $b$ and there is no mutation along the edges; 
or one of the founder lineages is of type $a$ ($b$) one is of type $j$ and a mutation occurs from  $j$ to $b$ ($a$) along the edge. The probability that $l=2$ and $n_a=1$, $n_b=1$ is therefore
\begin{eqnarray}
&&2x_ax_be^{-\theta t}e^{-t} + x_a^2\frac{\theta}{2}te^{-\theta t}2P_{ab}e^{-t} +\sum_{j\ne a}2x_ax_j\frac{\theta}{2}tP_{jb}e^{-\theta t}e^{-t} 
\nonumber \\
&&~~+x_b^2\frac{\theta}{2}te^{-\theta t}2P_{ba}e^{-t} +\sum_{j\ne b}2x_bx_j\frac{\theta}{2}tP_{ja}e^{-\theta t}e^{-t}  + o(\theta)
\nonumber \\
&&~=
2x_ax_b(1-\theta t)e^{-t}+\theta te^{-t} \sum_{j\in [d]} \left(x_ax_jP_{jb} + x_bx_jP_{ja} \right)+ o(\theta).
\label{xtwol:2}
\end{eqnarray}
If there is $l=1$ founder lineage then by an argument similar to that for the $n_a=2$ case it is necessary to consider any number of mutations on the single 
edge of the coalescent tree from time 0 to the coalescence event, and sufficient to consider at most one mutation event on the two edges from the coalescence event 
to the present time $t$.  The probability that $l=1$ and $n_a=1$, $n_b=1$ is therefore
\begin{eqnarray}
&& \sum_{j \in [d]} x_j \int_0^t \left(e^{(t - w)\gamma} \right)_{ja} e^{-\theta w} P_{ab} \theta w e^{-w} dw +  a \leftrightarrow b  + o(\theta) \nonumber \\
& = & \theta \sum_{i \in [d]} \left( \bm{x} e^{t\gamma} \right) _i  P_{ab} \int_0^t \left(e^{-w\gamma}\right)_{ia} w e^{-w} dw +  a \leftrightarrow b  + o(\theta) \nonumber \\ 
& = & \theta \left\{ \left( \bm{x} e^{t\gamma} \right) _a P_{ab} + \left( \bm{x} e^{t\gamma} \right) _b P_{ba} \right\} \left(1 - (1 + t)e^{-t} \right)  + o(\theta),  
\label{xtwol:1}
\end{eqnarray}
where $ a \leftrightarrow b $ indicates the previous term with $a$ and $b$ interchanged.  
The total probability of a configuration $n_a=1$, $n_b=1$ is the sum of Eqs.~(\ref{xtwol:2}) and  (\ref{xtwol:1}).  
Taking the limit $t \rightarrow \infty$ we recover the 
stationary distribution 
$$
p(\bm{e}_a + \bm{e}_b; \theta) = \theta(\pi_aP_{ab} + \pi_b P_{ba}) + o(\theta).  
$$
If $t << 1/\theta$, then Eqs.~ (\ref{xtwol:2}) and (\ref{xtwol:1}) give 
\begin{eqnarray} 
p(\bm{e}_a + \bm{e}_b; t, \bm{x}) & = & 2x_ax_b(1-\theta t)e^{-t}+\theta te^{-t} \sum_{j\in [d]} \left(x_ax_jP_{jb} + x_bx_jP_{ja} \right)\ +  \nonumber \\
&&  \qquad \theta(x_aP_{ab} + x_b P_{ba})\left(1 - (1 + t)e^{-t} \right) + o(\theta).  
\label{xtwol:smalltheta}
\end{eqnarray}
%
%
\begin{figure}[t]
\begin{center}
\centerline{\includegraphics[width=\textwidth]{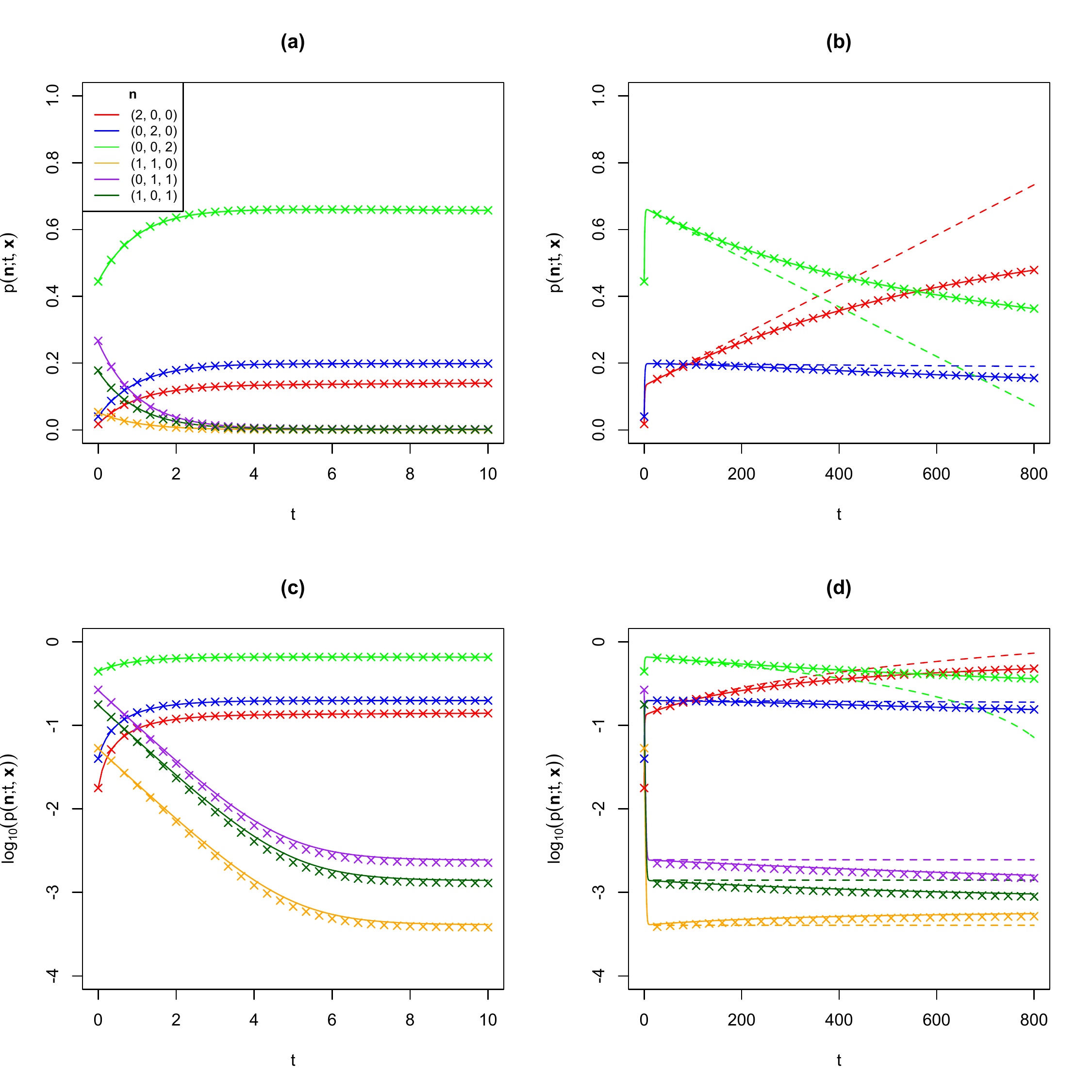}}
\caption{Comparison of ${\cal O}(\theta)$ calculations of the $n = 2$ sampling distributions $p(\bm{n}; t, \bm{x})$ (solid curves 
from Eqs.~(\ref{twol:2}), (\ref{twol:1}), (\ref{xtwol:2}) and (\ref{xtwol:1})); 
approximations to $p(\bm{n}; t, \bm{x})$ valid for $t << 1/\theta$ (dashed curves from Eqs.~(\ref{twol:smalltheta}) and (\ref{xtwol:smalltheta})); 
and exact numerically determined sampling distributions of the corresponding discrete-time, finite population neutral Wright-Fisher model with a population 
size $N = 15$.  The crosses correspond to every 5th time step of the Wright-Fisher model in (a) and (c), and every 400th time step in (b) and (d).  
Plots (a) and (b) are drawn on a linear scale; plots (c) and (d) are logarithmic to accentuate the $n_a = n_b = 1$ cases.  
The $t << 1/\theta$ plots are indistinguishable from the full ${\cal O}(\theta)$ calculation in (a) and (c).  Parameter values are as described in the text.  
}
\label{fig:PlotComparisonWithTheory}
\end{center}
\end{figure}
The above sampling distributions can be compared with an exact numerical calculation of the discrete-time, finite population, multi-allelic neutral Wright-Fisher model.   
Figure~\ref{fig:PlotComparisonWithTheory} shows the evolution of a $d = 3$ allele neutral Wright-Fisher model with a population size $N = 15$ 
and diffusion-limit instantaneous rate matrix 
$$
\gamma = \left( \begin{array}{rrr}
			-0.00056 &  0.00032 &  0.00024 \\
  			 0.00080 & -0.00360 &  0.00280 \\
  			 0.00100 &  0.00100 & -0.00200
		\end{array} \right),
$$
with corresponding 
$\theta = 2 \max_{1\leq i \leq d}(-\gamma_{ii}) = 0.00720$.
 The conventions used for defining the mapping between the discrete Wright-Fisher model and its diffusion limit 
are set out in \cite{BT2016}, Section~2.
The initial distribution corresponding to an initial allele frequency $\mathbf{x} = (2, 3, 10)/15$ was evolved 
over discrete time steps $\delta t = 1/N$ using the full $(N + 1)(N + 2)/2 = 136$ dimensional transition matrix defined in Eq.~(1) of \cite{BT2016}.  

As is generally observed in population genetics models, the diffusion limit is approached rapidly for moderate population sizes.  One observes that the 
${\cal O}(\theta)$ approximation to the sampling distribution agrees well with the discrete Wright-Fisher model over both short and long time scales for all six 
possible allele combinations with $n = \sum_a n_a = 2$.  On the other hand the approximation in Eqs.~(\ref{twol:smalltheta}) and (\ref{xtwol:smalltheta}) is only 
valid over short time scales as expected.  This illustrates the importance of allowing for multiple mutations when there is a single founder lineage over long time scales 
in the small $\theta$ limit.  
%
%
%
%
\section{Sampling Distribution for Arbitrary Sampling Size}
\label{sec:SamplingArbitraryN}

In a full $n$-coalescent tree the edge lengths $T_n,\cdots, T_2$ are independent exponential random variables with rates ${n\choose 2},\ldots, {2\choose 2}$, however in a non-stationary partial coalescent tree back from time $t$ the edge lengths are constrained by the origin. If there are $l$ founder lineages then it is sometimes necessary to use $q_{nl}(k;t) = \mathbb{E}\big [T_k\mathbb{I}\big \{A_n(t) = l\big \}\big ]$. Calculations of $q_{nl}(k;t)$ are made in Lemma \ref{lemma:1}.  
%
%
\begin{lemma}\label{lemma:1} Explicit forms for expected edge lengths $T_k$ in a Kingman coalescent tree conditioned on $l$ founder lineages at time 0 are, 
for $1 \leq l \leq k \leq n$ and $k \geq 2$,
\begin{eqnarray}
q_{nl}(k;t) &\equiv& \mathbb{E}\big [T_k\mathbb{I}\big \{A_n(t)=l\big \}\big ]
\nonumber \\
&=&\int_0^tq_{nl}(\tau)e^{-{k\choose 2}(t-\tau)}d\tau
\nonumber \\
&=&
\sum_{j=l}^n\omega_{jk}(t)(-1)^{j-l}
\frac{
(2j-1)(j+l-2)!
}
{
l!(l-1)!(j-k)!
}
\frac{n_{[j]}}{n_{(j)}},
\label{qnlk:00}
\end{eqnarray}
where
\[
\omega_{jk}(t) =
\begin{cases}
te^{-{k\choose 2}t}&\text{~if~}j=k\\
\frac{1}{{j\choose 2}-{k\choose 2}}
\Big (e^{-{k\choose 2}t} - e^{-{j\choose 2}t}\Big )
&\text{~if~}j\ne k.
\end{cases}
\]
and
\begin{eqnarray}
q_{n1}(1;t)
 &=& \int_0^tq_{n1}(\tau)d\tau
\nonumber \\
&=& \sum_{j=1}^n \int_0^t e^{-{j\choose 2}\tau}d\tau\ (-1)^{j-1}(2j-1)
\frac{n_{[j]}}{n_{(j)}}
\nonumber \\
&=& \sum_{j=1}^n 2\big (1-e^{-{j\choose 2}t}\big ) (-1)^{j-1}\frac{(2j-1)}{j(j-1)}
\frac{n_{[j]}}{n_{(j)}}
\label{qn1k:00}
\end{eqnarray}
\end{lemma}
%
%
\begin{proof}
Calculations evaluate the expectation
$
\mathbb{E}\big [T_k\mathbb{I}\big \{A_n(t) = l\big \}\big ]
$
 using the identity that
\[
\{A_n(t) = l\} \equiv \{T_n+\cdots +T_{l+1} <t < T_n+\cdots + T_l\},\>2 \leq l  < n.
\]
For $2 \leq l < k \leq n$
\begin{eqnarray}
&&\mathbb{E}\big [T_k\mathbb{I}\{T_n+\cdots +T_{l+1} <t < T_n+\cdots + T_l\}\big ]
\nonumber \\
&&~= \int_0^t\mathbb{E}\big [T_k\mathbb{I}\{T_n+\cdots +T_{l+1} = \tau\}\big ]
P(T_l > t-\tau)d\tau
\nonumber \\
&&~= \int_0^t\int_{t_i>0, i \in [n]\backslash[l];t_n+\ldots+t_{l+1}=\tau}t_k\prod_{j=l+1}^n{j\choose 2}e^{-{j\choose 2}t_j}dt_n\cdots dt_{l+1}e^{-{l\choose 2}(t-\tau)}d\tau.
\nonumber \\
\label{temp:50}
\end{eqnarray}
Eq.~(\ref{temp:50}) is a convolution of $n-l+1$ functions with a Laplace transform, for $\phi \geq 0$,
\begin{eqnarray*}
&&\int_0^\infty e^{-\phi t}\int_0^t\int_{t_i>0, i \in [n]\backslash[l];t_n+\ldots+t_{l+1}=\tau}t_k\prod_{j=l+1}^n{j\choose 2}e^{-{j\choose 2}t_j}dt_n\cdots dt_{l+1}e^{-{l\choose 2}(t-\tau)}d\tau dt
\nonumber \\
&&~=
{k\choose 2}^{-1}{l\choose 2}^{-1}
\Bigg ( 1 + {k\choose 2}^{-1}\phi\Bigg )^{-1} \prod_{j=l}^n\Bigg (1 + {j\choose 2}^{-1}\phi\Bigg )^{-1}.
\end{eqnarray*}
Therefore the integral is
\begin{eqnarray}
q_{nl}(k;t) ={l\choose 2}^{-1}\int_0^tf_{nl}(\tau)e^{-{k\choose 2}(t-\tau)}d\tau
=\int_0^tq_{nl}(\tau)e^{-{k\choose 2}(t-\tau)}d\tau.
\label{qnlk:0}
\end{eqnarray}
A modified argument is needed to show that the joint probability that $A_n(t) = l$ and mean length of edges while $l$ edges $t - \big (T_n+\cdots + T_{l+1}\big )$ 
is actually the same as Eq.~(\ref{qnlk:0}) when $k=l$. This expectation is
\begin{eqnarray*}
&&\int_0^t\mathbb{P}(T_n+\cdots +T_{l+1}=\tau)(t-\tau)\mathbb{P}(T_l > t-\tau)d\tau
\nonumber \\
&&=
 \int_0^t\int_{t_i>0, i \in [n]\backslash[l];t_n+\ldots,t_{l+1}=\tau}\prod_{j=l+1}^n{j\choose 2}e^{-{j\choose 2}t_j}dt_n\cdots dt_{l+1}(t-\tau)
 e^{-{l\choose 2}(t-\tau)}d\tau
 \nonumber \\
 &&= {l\choose 2}^{-1}\int_0^tf_{nl}(\tau)e^{-{l\choose 2}(t-\tau)}d\tau.
 \end{eqnarray*}
Therefore Eq.~(\ref{qnlk:0}) holds for $2 \leq l \leq k \leq n$. $q_{n1}(1;t)$ follows from definition.
\qed
\end{proof}
%
%

The probability of a sample configuration for $n\geq 2$ is now calculated in Theorem~\ref{theorem:1}.  Combinatorics are more involved because it is necessary to consider the sample configuration from $l \leq n$ founder genes.
%
%
\begin{theorem}\label{theorem:1}
The probability of a sample configuration $\bm{n}$ of $n$ individuals taken at time $t$ in the Wright-Fisher diffusion with generator in Eq.~(\ref{gen:0}) as $\theta \to 0$ is
\begin{eqnarray}
p(\bm{n};t,\bm{x})&=&
\sum_{2\leq l \leq n;|\bm{l}|=l}r_{\bm{n}\mid \bm{l}}m(\bm{l};l,\bm{x})
\Bigg (q_{nl}(t) - \frac{\theta}{2}\sum_{l \leq k \leq n} kq_{nl}(k;t)\Bigg )
\nonumber \\
&&+ \frac{\theta}{2}
\sum_{2\leq l \leq k \leq n;|\bm{l}|=l}
 q_{nl}(k;t)m(\bm{l};l,\bm{x})
 \nonumber \\
&&\times\Bigg (\sum_{i\ne j\in [d]; k_i,k_j; k_i+k_j \leq k}k_jP_{ji}\alpha_{ij}(k,k_i,k_j\mid l,l_i,l_j)\alpha_{ij}(n,n_i,n_j\mid k,k_i + 1,k_j - 1)r_{\bm{n}\backslash ij\mid\bm{l}\backslash ij}
\nonumber \\
&&~~~~+ \sum_{j\in [d]; k_j \leq k}k_jP_{jj}\alpha_{j}(k,k_j\mid l,l_j)
\alpha_{j}(n,n_j\mid k,k_j)r_{\bm{n}\backslash j\mid\bm{l}\backslash j} \Bigg )
\nonumber \\&&+
\sum_{a \in [d]}\mathbb{I}\{\bm{n}=n\bm{e}_a\}
\int_0^t(\bm{x}e^{\gamma (t-w)})_a \Big (q_{n2}(w) - \frac{\theta}{2}(1-P_{aa})\sum_{k=2}^nkq_{n2}(k;w)\Big )dw
\nonumber \\
&&+\frac{\theta}{2}\sum_{a\ne b \in [d]}\mathbb{I}\{\bm{n}=n_a\bm{e}_a+n_b\bm{e}_b\}
\nonumber \\
&&~~~\times\sum_{k=2}^nk\int_0^tq_{n2}(k;w)\Big (p_{nk}(n_b)(\bm{x}e^{\gamma (t-w)})_a
P_{ab} + p_{nk}(n_a)(\bm{x}e^{\gamma (t-w)})_b
P_{ba}\Big )dw + o(\theta),
\nonumber \\
\label{config:3a}
\end{eqnarray}
where 
$r_{\bm{n}\mid \bm{l}}, $ expressed in Eq.~(\ref{reqn:0}),
is the probability of a configuration $\bm{n}$ in a partial coalescent tree from a configuration $\bm{l} \leq \bm{n}$;
\begin{eqnarray*}
\alpha_{ij}(k,k_i,k_j\mid l,l_i,l_j) &=& {\cal DM}_{(l\backslash ij,l_i,l_j),k-l}
\big (k\backslash ij - l \backslash ij,k_i-l_i,k_j-l_j\big ),
\nonumber \\
\alpha_{j}(k,k_j\mid l,l_j) &=& {\cal DM}_{(l\backslash j,l_j),k-l}
\big (k\backslash j - l \backslash j,k_j-l_j\big ),
\end{eqnarray*}
with $q_{nl}(t) := \mathbb{P}\big (A_n(t) = l\big )$; for $1 \leq l \leq k \leq n$ and $k \geq 2$ (see Eq~(\ref{lod:0}));  and $q_{nl}(k;t) := \mathbb{E}\big [T_k\mathbb{I}\big \{A_n(t)=l\big \}\big ]$ with explicit forms in Lemma~\ref{lemma:1}. The Markov transition functions for mutation rate change are in the matrix $e^{\gamma t}$, where $\gamma = \frac{\theta}{2}(P-I)$; 
$m(\bm{l}, l, \bm{x})$ is the multinomial distribution Eq.~(\ref{multinomial}); and $p_{nk}(c)$ is defined in Eq.~(\ref{pnkc:0}).  
\end{theorem}
%
%
\begin{proof}

If there are no mutations occurring in the partial coalescent tree back to time 0 and the $l \geq 1$ founders have a configuration $\bm{l}$ then from a Polya urn argument considering $n-l$ branchings the probability of a configuration $\bm{n}$ at $t$ is
$r_{\bm{n}\mid \bm{l}}$.
The probability of a configuration of $\bm{n}$ when $l \geq 2$, and there is no mutation is 
\begin{eqnarray}
&&\sum_{2\leq l \leq n;|\bm{l}|=l}r_{\bm{n}\mid \bm{l}}m(\bm{l};l,\bm{x})
\mathbb{E}\Big [\mathbb{I}\big \{A_n(t) = l\big \}e^{-\frac{\theta}{2}\sum_{a=l}^naT_a}\Big ]
\nonumber \\
&&~=\sum_{2\leq l \leq n;|\bm{l}|=l}r_{\bm{n}\mid \bm{l}}m(\bm{l};l,\bm{x})
\Bigg (q_{nl}(t) - \frac{\theta}{2}\sum_{k=l}^n kq_{nl}(k;t)\Bigg )  + \mathcal{O}(\theta^2).
\label{config:11}
\end{eqnarray}
The probability of  a configuration $\bm{n}$, beginning from a configuration $\bm{l}$,  when there is one mutation from $j$ to $i\ne j$ while $k$ edges is
\begin{eqnarray}
&&\sum_{|\bm{k}|=k}
r_{\bm{k}\mid \bm{l}}
\mathbb{E}\Big [\frac{\theta}{2}k_jT_k\mathbb{I}\big \{A_n(t) = l\big \}e^{-\frac{\theta}{2}\sum_{a=l}^naT_a}\Big ]
P_{ji}r_{\bm{n}\mid \bm{k}+\bm{e}_i-\bm{e}_j}
\nonumber \\
&&~=\frac{\theta}{2}q_{nl}(k;t)\sum_{|\bm{k}|=k}
r_{\bm{k}\mid \bm{l}}k_jP_{ji}r_{\bm{n}\mid \bm{k}+\bm{e}_i-\bm{e}_j} + \mathcal{O}(\theta^2).
\label{temp:100}
\end{eqnarray}
The argument used to obtain Eq.~(\ref{temp:100}) is that the configuration from the $l$ founder edges to when $k$ edges; to after a mutation; then to when $n$ edges follows a Markov pattern $\bm{l}\to \bm{k} \to \bm{k} - \bm{e}_j + \bm{e}_i\to \bm{n}$.
The sum in Eq.~(\ref{temp:100}) can be simplified by summing over $\bm{k}$ keeping $k$, $k_i$ and $k_j$ fixed.
Now by using Eq.~(\ref{configuration:xx})
\begin{eqnarray}
r_{\bm{k}\mid\bm{l}} &=& \alpha_{ij}(k,k_i,k_j\mid l,l_i,l_j)
r_{\bm{k}\backslash ij\mid\bm{l}\backslash ij}
\nonumber \\
r_{\bm{n}\mid\bm{k}+\bm{e}_i - \bm{e}_j} &=&
\alpha_{ij}(n,n_i,n_j\mid k,k_i+1,k_j-1)
r_{\bm{n}\backslash ij\mid\bm{k}\backslash ij}.
\label{temp:101}
\end{eqnarray}
Summing over $\bm{k}\backslash ij$ with variables indexed by $i$ and $j$ held fixed,
using Eq.~(\ref{temp:101}),
\begin{eqnarray*}
&&\sum_{|\bm{k}\backslash ij |=k - k_i - k_j}
r_{\bm{k}\mid \bm{l}}r_{\bm{n}\mid \bm{k}+\bm{e}_i-\bm{e}_j}
\nonumber \\
&&~~~~~~
=\alpha_{ij}(k,k_i,k_j\mid l,l_i,l_j)\alpha_{ij}(n,n_i,n_j \mid k,k_i+1,k_j-1) r_{\bm{n}\backslash ij\mid\bm{l}\backslash ij}.
\end{eqnarray*}

Finally the case when $l=1$ where coalescence to a single ancestor occurs before $t$ is evaluated.
There is a need to consider mutations which may occur along the sample edges from the MRCA of the 
sample to the origin. If the MRCA occurs at time $w$ back from the current time then the type of the MRCA is determined by $\bm{x}e^{\gamma (t-w)}$. The probability of  a configuration $\bm{n}=n\bm{e}_a$ when $l=1$ is therefore
\begin{eqnarray}
&&\int_0^t(\bm{x}e^{\gamma (t-w)})_a
\mathbb{E}\big [\mathbb{I}\big \{A_n(w)=2\big \}\big (1+\frac{\theta}{2}\sum_{k=2}^nkT_kP_{aa}\big )e^{-\frac{\theta}{2}\sum_{k=2}^nkT_k}\big ]dw + \mathcal{O}(\theta^2)
\nonumber \\
&&~= \int_0^t(\bm{x}e^{\gamma (t-w)})_a
\Big ( q_{n2}(w) - \frac{\theta}{2}(1-P_{aa})\sum_{k=2}^nkq_{n2}(k;w)\Big ) dw+ \mathcal{O}(\theta^2),
\label{one:11}
\end{eqnarray}
arguing that the MRCA occurs at time $w$ back, the identity of the MRCA is $a$ and there are either no mutations, or one from $a$ to $a$, to the MRCA.

The probability of a configuration $\bm{n} = n_a\bm{e}_a + n_b\bm{e}_b$ similarly needs a consideration of the type of the MRCA being either $a$ or $b$, and a mutation from the MRCA type to the other type while $k$ lineages. This probability is
\begin{eqnarray}
&&\int_0^t(\bm{x}e^{\gamma (t-w)})_aP_{ab}
\mathbb{E}\big [\mathbb{I}\big \{A_n(w) = 2\}e^{-\frac{\theta}{2}\sum_{k=2}^nkT_k}\frac{\theta}{2}\sum_{k=2}^nkT_kp_{nk}(n_b)\big ]dw+ a \leftrightarrow b + \mathcal{O}(\theta^2)
\nonumber \\
&&~ = \frac{\theta}{2}\sum_{k=2}^nk\int_0^t(\bm{x}e^{\gamma (t-w)})_aP_{ab}q_{n2}(k;w)p_{nk}(n_b)dw + 
a \leftrightarrow b  + \mathcal{O}(\theta^2)
\label{two:22}
\end{eqnarray}

The total probability of a configuration of $\bm{n}$, Eq.~(\ref{config:3a}) in the statement of the theorem, is found from 
Eq.~(\ref{config:11}),  the simplified form of Eq.~(\ref{temp:100}), and Eqs.~(\ref{one:11}) and  (\ref{two:22}).
\qed
\end{proof}
%
%
%
%
\begin{corollary}\label{corollary:1}\citep{BG2019}\\
The stationary sampling distribution is
\begin{eqnarray}
\lim_{t\to \infty} p(\bm{n};t,\bm{x})
&=& \sum_{a\in [d]}\mathbb{I}\big \{\bm{n}=n\bm{e}_a\big \}\pi_a\big (1 - \theta (1-P_{aa})\big )\sum_{k=1}^{n-1}\frac{1}{k}\big )
\nonumber \\
&&+\sum_{a\ne b \in [d]}\mathbb{I}\{\bm{n}=n_a\bm{e}_a+n_b\bm{e}_b\}
\theta\big (\pi_aP_{ab}\frac{1}{n_b} + \pi_bP_{ba}\frac{1}{n_a}\big ) + \mathcal{O}(\theta^2).
\nonumber \\
\label{corr:11}	
\end{eqnarray}
\end{corollary}
%
%
\begin{proof} 
As $t \to \infty$ the MRCA is reached before $t$ with probability $1$, so only the last two terms in Eq.~(\ref{config:3a}) need to be considered. The time to the MRCA $w$ is finite and $(\bm{x}e^{\gamma (t-w)})_a \to \pi_a$.
The limit of the integrals in the term when $\bm{n}= n\bm{e}_a$ is
\begin{equation*}
\int_0^t (\bm{x}e^{\gamma (t-w)})_aq_{n2}(w)dw \to \pi_a\int_0^\infty q_{n2}(w)dw = \pi_a,
\end{equation*}
and
\begin{eqnarray*}
\frac{1}{2}\sum_{k=2}^nk	 \int_0^t(\bm{x}e^{\gamma (t-w)})_aq_{n2}(k;w)dw &\to& 
\frac{1}{2}\sum_{k=2}^nk	\pi_a\int_0^tq_{n2}(k;w)dw 
\nonumber \\
&\to&\frac{1}{2}\sum_{k=2}^nk\pi_a\int_0^\infty\mathbb{E}\big [T_k\mathbb{I}\big \{T_n+\cdots+T_2\in (w,w+dw)\big \}\big ]
\nonumber \\
&=& \frac{1}{2}\pi_a\sum_{k=2}^nk\mathbb{E}\big [T_k\big ]
= \pi_a\sum_{k=1}^{n-1}\frac{1}{k},
\end{eqnarray*}
giving the first term in Eq.~(\ref{corr:11}). The term in Eq.~(\ref{config:3a}) where $\bm{n}=n_a\bm{e}_a+n_b\bm{e}_b$ converges to
\begin{equation}
\theta\pi_aP_{ab}\sum_{k=2}^n\frac{1}{k-1}p_{nk}(n_b) + a \leftrightarrow b
= \theta \pi_aP_{ab}\frac{1}{n_b} + 	a \leftrightarrow b.
\label{temp:550}
\end{equation}
Simplification of the sum in Eq.~(\ref{temp:550}) is done in \citet{BG2019}.
\qed
\end{proof}
%
%
The combinatorial terms in Theorem~\ref{theorem:1} are much simpler if the initial frequency is $\bm{x}=\bm{e}_a$.
%
%
\begin{corollary} \label{corollary:x}
Suppose the initial ancestor types are all $a$, so $\bm{x}=\bm{e}_a$. As $\theta \to 0$ 
the probability of a sample configuration where all individuals are type $a$ is
\begin{eqnarray}
&&p(n\bm{e}_a;t,\bm{e}_a)
\nonumber \\
 &&~~
 =1 - q_{n1}(t) - \frac{\theta}{2}(1-P_{aa})\sum_{2 \leq l \leq k \leq n} kq_{nl}(k;t)
 \nonumber \\
 &&~+ \int_0^t (e^{\gamma(t-w)})_{aa}\Big ( q_{n2}(w) - \frac{\theta}{2}(1-P_{aa}) \sum_{k=2}^nkq_{n2}(k;w) \Big )dw
 +o(\theta);
\label{config:1000x}
\end{eqnarray}
the probability of a sample configuration with two types $a$ and $b$, $a \ne b$ is
\begin{eqnarray}
&&p(n_a\bm{e}_a+n_b\bm{e}_b;t,\bm{e}_a)
\nonumber \\
&&~~~~~~ =
\frac{\theta}{2}P_{ab}\sum_{2 \leq l \leq k \leq n}   kq_{nl}(k;t)p_{nk}(n_b)
\nonumber \\
&&~~~~~~~+\frac{\theta}{2}P_{ab}\int_0^t(e^{\gamma (t-w)})_{aa}\sum_{k=2}^nkq_{n2}(k;w)dw\ p_{nk}(n_b)
\nonumber \\
&&~~~~~~~+\frac{\theta}{2}P_{ba}\int_0^t(e^{\gamma (t-w)})_{ab}\sum_{k=2}^nkq_{n2}(k;w)dw\ p_{nk}(n_a)
 + o(\theta);
\label{config:00x}
\end{eqnarray}
the probability of a sample configuration with two types $b$ and $c$, $a \ne b \ne c \ne a$,  is
\begin{eqnarray}
&&p(n_b\bm{e}_b+n_c\bm{e}_c;t,\bm{e}_a)
\nonumber \\
&&~~~~~~ =
\frac{\theta}{2}P_{bc}\int_0^t(e^{\gamma (t-w)})_{ab}\sum_{k=2}^nkq_{n2}(k;w)dw\ p_{nk}(n_c)
\nonumber \\
&&~~~~~~+\frac{\theta}{2}P_{cb}\int_0^t(e^{\gamma (t-w)})_{ac}\sum_{k=2}^nkq_{n2}(k;w)dw\ p_{nk}(n_b)
 + o(\theta);
\label{config:00xa}
\end{eqnarray}
and the probability of a sample configuration with greater than two types is of 
$o(\theta)$.
\end{corollary}
%
%
\begin{proof}
The probability of a configuration $\bm{n}=n\bm{e}_a$ is
\begin{eqnarray*}
&&\sum_{l=2}^n\mathbb{E}\big [\mathbb{I}\big \{A_n(t)=l\big \}
e^{-\frac{\theta}{2}\sum_{k=l}^nkT_k}\big (1 + P_{aa}\frac{\theta}{2}\sum_{k=l}^nkT_k\big )\big ] 
\nonumber \\
&&+\int_0^t(e^{\gamma(t-w)})_{aa}\mathbb{E}\big [
\mathbb{I}\big \{A_n(w) = 2\big \}
\big (1 + P_{aa}\frac{\theta}{2}\sum_{k=2}^nkT_k\big ) e^{-\frac{\theta}{2}\sum_{k=2}^nkT_k}
\big ] dw
+ o(\theta)
\nonumber \\
\end{eqnarray*}
which is equal to Eq.~(\ref{config:1000x}) after expanding $e^{-\frac{\theta}{2}\sum_{k=2}^nkT_k}$ to $\mathcal{O}(\theta)$.
The probability of a configuration $\bm{n}=n_a\bm{e}_a+n_b\bm{e}_b$ is
\begin{eqnarray*}
&&\frac{\theta}{2}P_{ab}\sum_{l=2}^n\mathbb{E}\big [\mathbb{I}\big \{A_n(t)=l\big \}
\sum_{k=l}^n kT_k  p_{nk}(n_b) e^{-\frac{\theta}{2}\sum_{k=l}^n kT_k}\big ]
\nonumber \\
&&+\frac{\theta}{2}\int_0^t(e^{\gamma(t-w)})_{aa}\mathbb{E}\big [
\mathbb{I}\big \{A_n(w) = 2\big \}P_{ab}\sum_{k=2}^n kT_k p_{nk}(n_b)
e^{-\frac{\theta}{2}\sum_{k=2}^nkT_k}
\big ] dw
\nonumber \\
&&+\frac{\theta}{2}\int_0^t(e^{\gamma(t-w)})_{ab}\mathbb{E}\big [
\mathbb{I}\big \{A_n(w) = 2\big \}P_{ba}\sum_{k=2}^n kT_k p_{nk}(n_a)
e^{-\frac{\theta}{2}\sum_{k=2}^nkT_k}
\big ] dw
+ o(\theta),
\end{eqnarray*}
which is simplifies to Eq.~(\ref{config:00x}). 
A configuration $\bm{n}=n_b\bm{e}_b+n_c\bm{e}_c$, $b\ne c$, can only occur if the MRCA of the sample occurs before $t$ and the identity of the MRCA is $b$ or $c$. The probability of this sample configuration is 
\begin{eqnarray*}
&&\frac{\theta}{2}\int_0^t(e^{\gamma(t-w)})_{ab}\mathbb{E}\big [
\mathbb{I}\big \{A_n(w) = 2\big \}P_{bc}\sum_{k=2}^n kT_k p_{nk}(n_c)
e^{-\frac{\theta}{2}\sum_{k=2}^nkT_k}
\big ] dw
\nonumber \\
&&+\frac{\theta}{2}\int_0^t(e^{\gamma(t-w)})_{ac}\mathbb{E}\big [
\mathbb{I}\big \{A_n(w) = 2\big \}P_{cb}\sum_{k=2}^n kT_k p_{nk}(n_b)
e^{-\frac{\theta}{2}\sum_{k=2}^nkT_k}
\big ] dw
+ o(\theta),
\end{eqnarray*}
which simplifies to Eq.~(\ref{config:00xa}).
\qed
\end{proof}
%
%

One may be interested only in time scales $t << 1/\theta$, in which case the following Corollaries~\ref{corollary:fixedt} and \ref{corollary:4} are relevant.  
If $t$ is fixed and $\theta \to 0$ then the probability of greater than one mutation along the edge from the MRCA to to the origin is $\mathcal{O}(\theta^2)$ and the last 
two terms in Eq.~(\ref{config:3a}) can be simplified.
%
%
\begin{lemma}\label{lemma:2} The following identities relating to the case $l = 2$ will be of use in Corollaries~\ref{corollary:fixedt} and \ref{corollary:4}:  
\begin{equation*}
\begin{split}
\int_0^t q_{n2}(w)dw &= q_{n1}(t),\\
\int_0^t q_{n2}(k;w)dw &= q_{n1}(k;t), \\
\int_0^t(t-w)q_{n2}(w)dw &= q_{n1}(1;t).
\end{split}
\end{equation*}
\end{lemma}
%
%
\begin{proof}
From Eq.~(\ref{identity:100}) with $l = 2$, 
$\int_0^t q_{n2}(x)dw  =  \int_0^t f_{n2}(x)dw
	 =  \mathbb{P}(T_n + \cdots +T_2 < t) 
	 =  \mathbb{P}(A_{n1}(t) = 1) 
	 =  q_{n1}(t)$.
Alternatively one can differentiate Eq.~(\ref{lod:0}) with $l = 1$, and set the constant of integration by noting that $q_{n1}(0) = 0$ for $n \ge 2$.  
For the second identity
$q_{n1}(k;t) = \mathbb{E}\big [T_kI\big \{A_n(t)=1\big \} \big ]
= \int_0^t\mathbb{E}\big [T_kI\big \{T_n+\cdots + T_2 \in (w,w+dw)\big \}\big ]
= \int_0^t\mathbb{E}\big [T_kI\big \{A_n(w)=2\big \}\big ]dw = \int_0^tq_{n2}(k;w)dw$.
The last identity is proved by noting that $\int_0^tq_{n2}(w)dw = q_{n1}(t)$, so $q_{n2}(t) = q^\prime_{n1}(t)$ and
$\int_0^t(t-w)q_{n2}(w)dw = \int_0^t(t-w)q^\prime_{n1}(w)dw = (t-w)q_{n1}(w)\Big ]_0^t + \int_0^tq_{n1}(w)dw = 0 + q_{n1}(1;t)$.
\qed
\end{proof}
%
%

%
%
\begin{corollary}\label{corollary:fixedt}
If $t << 1/\theta$ as $\theta \to 0$ the last two terms in Eq.~(\ref{config:3a}) become
\begin{eqnarray}
&&\sum_{a \in [d]}\mathbb{I}\{\bm{n}=n\bm{e}_a\}\Biggl (x_a\Big (
q_{n1}(t) - \frac{\theta}{2}(1-P_{aa})\sum_{k=1}^nkq_{n1}(k;t)
\Big ) + \frac{\theta}{2}\sum_{j\in [d], j \ne a}x_jP_{ja})\Biggr )
\nonumber \\
&&+\frac{\theta}{2}\sum_{a\ne b \in [d]}\mathbb{I}\{\bm{n}=n_a\bm{e}_a+n_b\bm{e}_b\}
\sum_{k=2}^nkq_{n1}(k;t)\big (p_{nk}(n_b)x_aP_{ab} + p_{nk}(n_a)x_bP_{ba}\big )
+ \mathcal{O}(\theta^2),
\nonumber \\
\label{xyz:500}	
\end{eqnarray}
where expressions in Eq.~(\ref{xyz:500}) are defined previously.
\end{corollary}
%
%
\begin{proof}
In Eqs.~(\ref{one:11}) and  (\ref{two:22}) expand 
$\bm{x}e^{\gamma(t-w)} = \bm{x}\big (I + \gamma(t-w)\Big ) + \mathcal{O}(\theta^2)$.
Then  Eq.~(\ref{one:11}) becomes 
\begin{equation*}
\int_0^t \Bigl (x_a + \frac{\theta}{2}\sum_{j\in [d]}x_j(P_{ja}-\delta_{ja})(t-w)\Bigr )
\Bigl (q_{n2}(w) - \frac{\theta}{2}(1-P_{aa})\sum_{k=2}^nkq_{n2}(k;w)\Bigr )dw+\mathcal{O}(\theta^2),
\end{equation*}
which simplifies to Eq.~(\ref{xyz:500}) after using the identities in Lemma~\ref{lemma:2}.  
\qed
\end{proof}
%
%

%
%
\begin{corollary} \label{corollary:4}
Suppose the initial ancestor types are all $a$, so $\bm{x}=\bm{e}_a$. 
If $t<<1/\theta$ the probability of a sample configuration where all individuals are type $a$ is
\begin{eqnarray}
p(n\bm{e}_a;t,\bm{e}_a) &=& 
 1 - \frac{\theta}{2}(1-P_{aa})\sum_{1 \leq l \leq k \leq n} kq_{nl}(k;t)
 +\mathcal{O}(\theta^2);
\label{config:1000}
\end{eqnarray}
the probability of a sample configuration with two types $a$ and $b$ is
\begin{eqnarray}
p(n_a\bm{e}_a+n_b\bm{e}_b;t,\bm{e}_a) =
\frac{\theta}{2}P_{ab}\sum_{l=1}^n\sum_{k\geq l,2}^nkq_{nl}(k;t)p_{nk}(n_b) + \mathcal{O}(\theta^2);
\label{config:00a}
\end{eqnarray}
and the probability of a sample configuration with greater than two types is of 
$\mathcal{O}(\theta^2)$.
\end{corollary}
%
%
\begin{proof}
The last line of Eq.~(\ref{config:1000x}) is equal to
\begin{eqnarray*}
&&q_{n1}(t) + \frac{\theta}{2}(P_{aa}-1)\int_0^t(t-w)q_{n2}(w)dw
- \frac{\theta}{2}(1-P_{aa})\sum_{k=2}^nkq_{n1}(k;t) + \mathcal{O}(\theta^2)
\nonumber \\
&&~= q_{n1}(t) - 	\frac{\theta}{2}(1-P_{aa})\sum_{k=1}^nkq_{n1}(k;t) + \mathcal{O}(\theta^2).
\end{eqnarray*}
Now add the first line of Eq.~(\ref{config:1000x}) to obtain Eq.~(\ref{config:1000}).
The probability of a configuration $\bm{n}=n_a\bm{e}_a+n_b\bm{e}_b$, Eq.~(\ref{config:00a}), is found by noting that 
$
(e^{\gamma(t-w)})_{ab} =\delta_{ab} + \mathcal{O}(\theta)
$
so the second last line of Eq.~(\ref{config:00x}) is 
\begin{equation}
\frac{\theta}{2}P_{ab}\sum_{k=2}^nkq_{n1}(k;t) p_{nk}(n_b)+\mathcal{O}(\theta^2)	
\label{config:xyz}
\end{equation}
and the last line is $\mathcal{O}(\theta^2)$. Then Eq.~(\ref{config:00a}) follows by adding Eq.~(\ref{config:xyz}) 
and the first line of Eq.~(\ref{config:00x}).   
\qed
\end{proof}
\section{The infinitely-many-sites model and gene trees}
In the infinitely-many-sites model genes are sequences. Mutations occur at positions never before mutant in the population of sequences. The configuration of mutations at segregating sites in a sample is equivalent to a unique gene tree, called a perfect phylogeny in biology. A gene tree of $n$ sequences of $g\leq d$ different types is constructed by labeling mutations in the coalescent tree of the $n$ genes, then describing each gene by its mutation path $\bm{z}$ from leaves to the root of the tree with multiplicities of the types $\bm{n}=(n_1,\ldots,n_g)$. The same ideas used in Theorem 1 lead to an approximate probability for a gene tree when $\theta$ is small. It is natural to consider gene trees because they are closely related to coalescent trees. The gene tree does not depend on the actual values of the labels, only their pattern.
This is a very brief account; the reader is referred to \citet{EG1997,G1999,GT1999} for more details.
An illustration of a coalescent tree showing these concepts is in Fig. \ref{fig:ctree}.
\medskip

\begin{figure}
\begin{center}
\caption{Example coalescent tree}\label{fig:ctree}
\bigskip

\begin{tikzpicture}[xscale=0.7,yscale = 0.3]
\draw (0,0) -- (0,4);
\draw (1,0) -- (1,4);
\draw (2,0) -- (2,6);
\draw (0.5,4)  -- (0.5,6);
\draw (1.25,6) -- (1.25,7);
\draw (3,0) -- (3,7);
\draw (4,0) -- (4,7);
\draw (5,0) -- (5,2);
\draw (6,0) -- (6,2);
\draw (5.5,2) -- (5.5,7);
\draw (7,0)-- (7,5);
\draw (8,0) -- (8,5);
\draw (7.5,5) -- (7.5,7);
\draw [dashed] (1.25,7) -- (1.25,13);
\draw [dashed] (3,7) -- (3,8);
\draw [dashed] (4,7) -- (4,8);
\draw [dashed] (3.5,8) -- (3.5,10);
\draw [dashed] (5.5,7) -- (5.5,10);
\draw [dashed] (4.5,10) -- (4.5,13);
\draw [dashed] (2.75,13) -- (2.75,15);
\draw [dashed] (7.5,7) -- (7.5,15);
\draw [dashed] (5.125,15) -- (5.125,17);
\draw (0,4) -- (1,4);
\draw (0.5,6) -- (2,6) ;
\draw (5,2) -- (6,2);
\draw (7,5) -- (8,5);
\draw [dashed] (1.25,13) -- (4.5,13);
\draw [dashed] (3.5,10) -- (5.5,10);
\draw [dashed] (2.75,15) -- (7.5,15);
\draw [dashed] (3,8) -- (4,8);
\draw -- (10,7) node{time $0$};
\draw -- (10,0) node {time t};
\draw -- (5.125,16)  node {$\bullet$};
\draw -- (5.625,16) node {$z_3$};
\draw -- (4.5,11)  node {$\bullet$};
\draw -- (5,11) node {$z_2$};
\draw -- (5.5,3)  node {$\bullet$};
\draw -- (6,3) node {$z_0$};
\draw -- (1.25,10)  node {$\bullet$};
\draw -- (1.75,10) node {$z_1$};
\end{tikzpicture}
\end{center}
\bigskip

The initial coalescent tree at time zero is shown as the dashed subtree. The full coalescent tree at time $t$ is the tree with dashed and non-dashed lines.
The initial gene tree with $l=5$ founder lineages is
$
T(\bm{l}) = \big ( (z_1,z_3),(z_2,z_3),(z_3)\big )
$
with multiplicities of the sequences $\bm{l}=(1,3,1)$. There is a mutation labeled $z_0$ in $(0,t)$ and the gene tree at time $t$
is
$
T = \big ((z_1,z_3),(z_2,z_3), (z_3), (z_0,z_2,z_3)\big )
$
in a sample of $n=9$ with multiplicities $\bm{n}=(3,2,2,2)$.
%
\end{figure}
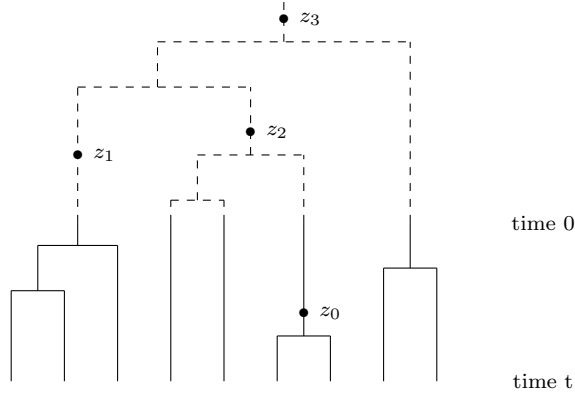
\bigskip

An analog of the result in Theorem 1 for gene trees is
\begin{eqnarray}
p\big ((T,\bm{n});t,\bm{x}\big )&=&
\sum_{2\leq l \leq n;|\bm{l}|=l;T(\bm{l}) = T}r_{\bm{n}\mid \bm{l}}m(\bm{l};l,\bm{x})
\Bigg (q_{nl}(t) - \frac{\theta}{2}\sum_{l \leq k \leq n} kq_{nl}(k;t)\Bigg )
\nonumber \\
&&+ \frac{\theta}{2}
\sum_{2\leq l \leq k \leq n;|\bm{l}|=l}
 q_{nl}(k;t)m(\bm{l};l,\bm{x})
 \nonumber \\
&&~~~~~~~~~~~~~~~~\times
\Bigg (\sum_{\bm{k}, |\bm{k}| = k; j\in [g];T_j(\bm{l})=T}
k_jr_{\bm{k}\mid\bm{l}}r_{\bm{n}\mid \bm{k}-\bm{e}_j + \bm{e}_*}\Bigg )
\nonumber \\&&+
\mathbb{I}\{T=\big((z_0)\big );\bm{n}=n\bm{e}_1\}
\Big (q_{n1}(t) - \frac{\theta}{2}\sum_{k=2}^nkq_{n1}(k;t)\Big )
\nonumber \\
&&+\frac{\theta}{2}\mathbb{I}\{T= \big ((z_1),(z_0,z_1)\big ),\bm{n}=n_1\bm{e}_1+n_2\bm{e}_2\}
\nonumber \\
&&~~~\times\sum_{k=2}^nkp_{nk}(n_2)q_{n1}(k;t)
+ o(\theta).
\label{config:gt}
\end{eqnarray}
The initial population of sequences has frequencies $\bm{x}$ which may have an infinite number of entries. In the labeling of $\bm{l}$ sample sequences only the non-zero number of sequences are included and relabeled.

The first line of (\ref{config:gt}) corresponds to the case where no mutation occurs in $(0,t)$.  In this case the gene trees at time zero and time $t$, specified as above in terms of paths to the 
root, are equal, that is $T(\bm{l}) = T$.  However the trees have different multiplicities $\bm{l}$ and $\bm{n}$, and the argument leading to Eq.(\ref{config:11}) applies. 
In the second line one mutation occurs adding a new singleton sequence to the tree.
$T_j(\bm{l})$ is the gene tree $T(\bm{l})$ modified by mutation to one of the $j$th sequences having an additional mutation label added. In the sum over $\bf{k}$, there are two cases to consider;
if $k_j>1$ there are still representatives of the $j$th sequence; and if $k_j=1$ the original sequence is removed. The new sequence is a singleton. Its position, denoted by $*$ is either at the end of the sequences, or in the position of the prior $j$th sequence which was removed in the second case. 

The third and fourth lines correspond to cases where the coalescent tree has a MRCA in $(0, t)$.   Then either the gene tree is a single sequence $\big ( (z_0)\big )$ of $n$ copies 
at $t$ if there is no further mutation, or there is a mutation and then the gene tree at $t$ is  $T= \big ((z_1),(z_0,z_1)\big )$ with multiplicities $(n_1,n_2)$.
In the third line $T= \big ( (z_0)\big )$ is a tree which is the same type as a single gene chosen from the initial sequences $\bm{x}$.
In the fourth line $T= \big ((z_1),(z_0,z_1)\big )$ does not distinguish whether or not $z_1$ is from the initial population, or a new mutation.  If there is no mutation along the single lineage it will be from the initial population, or if there is mutation it will denote a new mutation. To split the terms up into these two cases replace $q_{n1}(k;t)$
by $\int_0^tq_{n2}(k;w)e^{-\frac{\theta}{2}(t-w)}dw$ for no mutation and $\int_0^tq_{n2}(k;w)\big (1-e^{-\frac{\theta}{2}(t-w)}\big )dw$ if mutation occurs.

The infinitely-many-alleles model is a model of gene frequencies where mutation always produces a new allele type. The distribution of the configuration of types in a sample of $n$ genes is the Ewens' sampling formula \citep{E1972}. The infinitely-many-alleles model is contained in the infinitely-many-sites model by considering the first co-ordinates of the sequences in the gene trees.  The two models can be thought of at a macroscopic and microscopic level of detail in DNA sequences. There is a similar equation to (\ref{config:gt}) for the infinitely-many-alleles model, keeping the types labeled for convenience.

%
\section{An approximate transition density for large sample sizes}
Denote by $f(\bm{x},\bm{y};t,\theta)$ the transition density of the Wright-Fisher diffusion process $\{\bm{X}(t)\}_{t\geq 0}$ with initial frequencies $\bm{X}(0) = \bm{x}$ 
and final frequencies $\bm{X}(t) = \bm{y}$. 
The transition density $f(\bm{x},\bm{y},t) \equiv f(\bm{x},\bm{y};t,0)$ in a model with no mutation. 
For functions $g(\bm{X}(t))$ with finite expectation with $\theta \geq 0$ 
\[
\lim_{\theta \to 0}\mathbb{E}^\theta_{\bm{x}}\big [g(\bm{X}(t))\big ] = \mathbb{E}^0_{\bm{x}}\big [g(\bm{X}(t))\big ],
\]
where $\mathbb{E}^\theta_{\bm{x}}$ denotes expectation with parameter $\theta$ and initial frequencies $\bm{x}$.

An expression for the transition density as a mixture of Dirichlet distributions is known from \citet{EG1993} and discussed in the review of \citet{GS2010}.  A description of the structure of the population at $t$ and this transition density is now given.  The population at time $t$ coalesces back to a random number of lineages $l$ at time 0, with probability $q_{\infty l}(t)$. The population is composed of families of individuals from these $l$ founders, with family sizes having an $l$ dimensional Dirichlet distribution with parameters $(1,\cdots,1)$. Family types are determined by the type of the founder lineages, so there is a $d$-dimensional $m(\bm{l},l,\bm{x})$ distribution of the accumulated number of families with types in $[d]$. The probability that in a subset $A \subseteq [d]$, $X_j(t) >0$ for $j\in A$ and $X_j(t)=0$ for $j \not\in A$ is the probability the set of founder lineage types is $A$.
The transition density is then
\begin{equation}
f(\bm{x},\bm{y};t) = \sum_{l=1}^\infty q_{\infty l}(t)
\sum_{|\bm{l}|=l}m(\bm{l},l,\bm{x}){\cal D}_{\bm{l}}(\bm{y}),
\label{transition:0}
\end{equation}
with the convention in the Dirichlet distribution that $y_j=0$ if $l_j=0$.
That is
\[
{\cal D}_{\bm{l}}(\bm{y})=\frac{\Gamma(|\bm{l}|)}{\prod_{j:l_j>0}\Gamma(l_j)}\prod_{j:l_j>0}y_j^{l_j-1}.
\]

 Although the first summation is formally written beginning at 1, because of the convention that all the terms with $l < |A|$ are zero the summation really begins with $l=|A|$.
 See \citet{EG1993} Eq.~(1.27) for a more general expansion which includes Eq.~(\ref{transition:0}) when their $\theta=0$.
A first example: if $A=\{1,2\}$, so $y_1>0,y_2>0$ and $y_1+y_2=1$, the transition density is
\begin{equation}
f\big (\bm{x},(y_1,y_2,\bm{0});t \big )
= \sum_{l=2}^\infty 
q_{\infty l}(t)
\sum_{l_1,l_2>1:l_1+l_2=l}m((l_1,l_2,\bm{0}),l,\bm{x})
\frac{\Gamma (l)}{\Gamma (l_1)\Gamma (l_2)}y_1^{l_1-1}y_2^{l_2-1},
\end{equation}
which is the joint density of $X_1(t),X_2(t)$ and the probability that $X_j(t) = 0, j > 2$. A second example: if $\bm{y}=\bm{e}_1$ the probability the population is fixed at time $t$ with type 1 individuals is the probability that all founder lineages are of type 1,
\[
f\big (\bm{x},\bm{e}_1;t \big )
= \sum_{l=1}^\infty q_{\infty l}(t)x_1^l.
\]

It would be ideal to find an expansion of $f(\bm{x},\bm{y};t,\theta)$ in terms of $f(\bm{x},\bm{y};t,0)$ to the first order in $\theta$ similar to the sample transition distribution, however for fixed $\theta$ the probability of greater than one mutation in an infinite-leaf coalescent tree is one, not of order $\theta^2$. 
The error for a probability configuration in a sample of $n$ by omitting the probability of greater than one mutation (denoted by $Q_1$) is bounded by the probability of greater than one mutation in an unconstrained coalescent tree.

This probability is exactly
\begin{eqnarray}
Q_1&=&\mathbb{E}\Big [1 - e^{-(1/2)\theta\sum_{l=2}^nlT_l}
\Big ( 1 + (\theta/2)\sum_{k=2}^nkT_k\Big )\Big ]
\nonumber \\
&=&1 - \prod_{l=1}^{n-1}\Big ( 1 + \frac{\theta}{l}\Big )^{-1}
\Big ( 1 + \theta\sum_{k=1}^{n-1}\frac{1}{k}\Big ( 1 + \frac{\theta}{k}\Big )^{-1}\Big )
\nonumber \\
&=&1 - \prod_{l=1}^{n-1}\Big ( 1 + \frac{\theta}{l}\Big )^{-1}
\Big ( 1 + \theta\sum_{k=1}^{n-1}\frac{1}{k+\theta}\Big )
\nonumber \\
&=& -\Bigg (\prod_{l=1}^{n-1}\Big ( 1 + \frac{\theta}{l}\Big )^{-1} - 1 + \theta\sum_{k=1}^{n-1}\frac{1}{k}\Bigg )\Big (1 + \theta\sum_{k=1}^{n-1}\frac{1}{k+\theta}\Big )
\nonumber \\
&&~~~~~+ \theta^2\Bigg(\sum_{k=1}^{n-1}\frac{1}{k(k+\theta)} + \sum_{k,l=1}^{n-1}\frac{1}{k(l+\theta)}\Bigg )
\nonumber \\
&=& -\Bigg (\theta\frac{(n-1)!\Gamma(\theta)}{\Gamma(n+\theta)} 
- 1 + \theta\sum_{k=1}^{n-1}\frac{1}{k}\Bigg )\Big (1 + \theta\sum_{k=1}^{n-1}\frac{1}{k+\theta}\Big )
\nonumber \\
&&~~~~~+ \theta^2\Bigg(\sum_{k=1}^{n-1}\frac{1}{k(k+\theta)} + \sum_{k,l=1}^{n-1}\frac{1}{k(l+\theta)}\Bigg )
\label{maxbound:0}
\end{eqnarray}
Both terms in (\ref{maxbound:0}) are of order $\theta^2$. To see this for the first term note that
\begin{eqnarray*}
\prod_{l=1}^{n-1}\Big ( 1 + \frac{\theta}{l}\Big )^{-1}
& =& 
\prod_{l=1}^{n-1}\Big (1 - \frac{\theta}{l} + \mathcal{O}(\theta^2)\Big )\\
&=&
1 - \theta\sum_{k=1}^{n-1}\frac{1}{k} + \mathcal{O}(\theta^2).
\end{eqnarray*}
Alternatively use the properties of the gamma function that as $\theta \to 0$
\begin{eqnarray*}
\Gamma(\theta) &=& \frac{1}{\theta} - \gamma + \mathcal{O}(\theta)\\
\Gamma(\theta+n) &=& \big (1 + \theta\psi(n)\big)(n-1)! + \mathcal{O}(\theta^2),
\end{eqnarray*}
where $\gamma$ is Euler's constant and $\psi(\cdot)$ is the diagamma function
with
\[
\psi(n) = \sum_{k=1}^{n-1}\frac{1}{k} - \gamma
\]
to obtain the result.

As $n\to \infty$ and $\theta\to 0$, there are three possible limits for $Q_1$, seen from the second line of (\ref{maxbound:0});
\begin{equation}
\lim_{\theta\log n \to 0}Q_1=0, \>
\lim_{\theta\log n \to \alpha}Q_1=1 - e^{-\alpha}(1+\alpha),\>
\lim_{\theta\log n \to \infty}Q_1= 1.
\label{Q1forAlpha}
\end{equation}
If $n\to \infty$ and $\theta \to 0$ such that $\theta\log n\to 0$ then the probability of greater than one mutation in an $n$-coalescent tree is $o(\theta)$; we take this approach and find an approximate transition density $f_n(\bm{x},\bm{y};t,\theta)$.
\begin{theorem}\label{theorem:3}
An approximate transition density in the general mutation Wright-Fisher diffusion from the sample transition distribution as $n\to \infty$ and $\theta \to 0$ such that $\theta \log n\to 0$, and
$\bm{n}/n \to \bm{y}$, with initial relative frequencies $\bm{x}$ is given by
\begin{eqnarray}
&&f_n(\bm{x},\bm{y};t,\theta) =
\sum_{2\leq l ;|\bm{l}|=l}m(\bm{l};l,\bm{x}){\cal D}_{\bm{l}}(\bm{y})
\Bigg (q_{n l}(t) - \frac{\theta}{2}\sum_{k=l}^n kq_{nl}(k;t)\Bigg )
\nonumber \\
&&+ \frac{\theta}{2}
\sum_{2\leq l \leq k \leq n;|\bm{l}|=l}
 q_{nl}(k;t)m(\bm{l};l,\bm{x})
 \nonumber \\
&&\times\Bigg (\sum_{i\ne j\in [d]; k_i,k_j; k_i+k_j \leq k}k_jP_{ji}
\alpha_{ij}(k,k_i,k_j\mid l,l_i,l_j)
\nonumber \\
&&~~~~\times{\cal D}_{ (k\backslash ij,k_i+1,k_j-1)}(1-y_i-y_j,y_i,y_j)
\cdot {\cal D}_{\bm{l}\backslash ij}(\bm{y}\backslash ij)
\nonumber \\
&&~~~~+ \sum_{j\in [d]; k_j \leq k}k_jP_{jj}
\alpha_{j}(k,k_j\mid l,l_j)\cdot
{\cal D}_{(k-k_j,k_j)}(1-y_j,y_j){\cal D}_{\bm{l}\backslash j}(\bm{y}\backslash j)\Bigg)
\nonumber \\
&&+\sum_{a \in [d]}\mathbb{I}\{\bm{y}=\bm{e}_a\}
\int_0^t(\bm{x}e^{\gamma (t-w)})_a \Big (q_{n 2}(w) - \frac{\theta}{2}(1-P_{aa})\sum_{k=2}^n kq_{n2}(k;w)\Big )dw
\nonumber \\
&&+\frac{\theta}{2}\sum_{a\ne b \in [d]}\mathbb{I}\{\bm{y}=y_a\bm{e}_a+y_b\bm{e}_b\}
\nonumber \\
&&~~~\times\sum_{k=2}^n k(k-1)\int_0^tq_{n 2}(k;w)\Big ((1-y_b)^{k-2}(\bm{x}e^{\gamma (t-w)})_a
P_{ab}  
+ a \leftrightarrow b \Big ) dw
+ o(\theta).
\nonumber \\
\label{config:300}
\end{eqnarray}
\end{theorem}
\begin{proof}
The proof follows immediately by taking $n \to \infty$ such that $\bm{n}/n \to \bm{y}$ in $p(\bm{n};t,\bm{x})$ of Theorem~\ref{theorem:1}. 
The transition density for $\bm{X}(t)$ can include corners of the space where some variables are zero.
Inversion of the sampling distributions in the expansion Eq.~(\ref{config:3a}) follows from;
\begin{eqnarray}
r_{\bm{n}\mid\bm{l}} &=& \int_{\Delta_d} m(\bm{n};n,\bm{y}){\cal D}_{\bm{l}}(\bm{y}) d\bm{y};
\nonumber \\
\alpha_{ij}(n,n_i,n_j | k,k_i + 1,k_j -1) &=& {\cal DM}_{(k\backslash ij,k_i+1,k_j-1),n-k}\big (n\backslash ij -k\backslash ij ,n_i-k_i-1,n_j-k_j+1\big )
\nonumber \\
&=& \int_{0 < y_i+y_j < 1}
\frac{n!}{(n\backslash ij)!n_i!n_j!}y_i^{n_i}y_j^{n_j}(1-y_i-y_j)^{n\backslash ij}
\nonumber \\
&&\times \frac{\Gamma (k)}{\Gamma (k\backslash ij)\Gamma (k_i+1) \Gamma (k_j+1)} 
y_i^{k_i}y_j^{k_j-2}(1-y_i-y_j)^{k\backslash ij-1}dy_idy_j;
\nonumber \\
&=& \int_{0 < y_i+y_j < 1}
m((n\backslash ij, n_i, n_j); n, (1 - y_i - y_j, y_i, y_j)) 
\nonumber \\
&& \qquad \times {\cal D}_{ (k\backslash ij,k_i+1,k_j-1)}(1-y_i-y_j,y_i,y_j)dy_idy_j;
\nonumber \\
p_{nk}(c) &=&
\int_0^1{n-k\choose c-1}z^{c-1}(1-z)^{n-c-k+1}\times (k-1)(1-z)^{k-2}dz.
\end{eqnarray}
The generalised Dirichlet distributions ${\cal D}_{\bm{l}}(\bm{y})$, ${\cal D}_{\bm{l}\backslash ij}(\bm{y}\backslash ij)$ have the property that if variables $\bm{y}$ or $\bm{y}\backslash ij$ are  zero, the resulting  Dirichlet distribution is of lower dimension than $d$.
The transition density can include corners of the space where some of the variables are zero if some entries in $\bm{x}$ are zero.

The condition that $\theta\log n \to 0$ is needed to ensure that as $\theta \to 0$, $n \to \infty$ then $\frac{\theta}{2}\sum_{k=l}^nkq_{nl}(k;t) \to 0$ as well as other similar terms in (\ref{config:300}). This is so because
\[
\frac{\theta}{2}\sum_{k=l}^nkq_{nl}(k;t)\leq \frac{\theta}{2}\sum_{k=2}^nk\mathbb{E}\big [T_k\big ] = \theta\sum_{k=2}^n\frac{1}{k-1} \to 0.
\]
\end{proof}

In \citet{BT2016,BG2019} the way a stationary distribution is developed is by finding that if $\theta$ is small, then the population is approximately fixed or contains just two allele types $a$, $b$ with relative frequencies on a line $y_a+y_b=1$.
This is a boundary approximation.
The transition density (\ref{config:300}) when there are two types is 
\begin{equation}
\frac{\theta}{2}\sum_{k=2}^\infty k(k-1)\int_0^tq_{\infty 2}(k;w)\Big ((1-y_b)^{k-2}(\bm{x}e^{\gamma (t-w)})_a
P_{ab} 
+ a\leftrightarrow b
\Big )dw.
\label{twotypes:0}
\end{equation}
The sum in (\ref{twotypes:0}) converges replacing $n$ by $\infty$ from (\ref{config:300}) because for $0 < y_b < 1$ (considering the first term)
\begin{eqnarray*}
&&\sum_{k=2}^\infty k(k-1)(1-y_b)^{k-2}\int_0^tq_{\infty 2}(k;w)(\bm{x}e^{\gamma (t-w)})_bdw
\\
&&~\leq \sum_{k=2}^\infty k(k-1)(1-y_b)^{k-2}\int_0^tq_{\infty 2}(k;w)dw
\\
&&~= \sum_{k=2}^\infty k(k-1)(1-y_b)^{k-2}q_{\infty 1}(k;t)
\\
&&~\leq \sum_{k=2}^\infty k(k-1)(1-y_b)^{k-2}\mathbb{E}\big [T_k\big]
= 2y_b^{-1} < \infty.
\end{eqnarray*}
The limit stationary distribution in (\ref{twotypes:0}) as $t \to \infty$ is
\[
\theta \big (\pi_aP_{ab}y_b^{-1} + \pi_bP_{ba}y_a^{-1}\big ),
\]
on the line $y_a+y_b=1$, as in \citet{BT2016,BG2019},
because $(e^{\gamma (t-w)})_a \to \pi_a$ and 
$\int_0^\infty q_{\infty 2}(k;w)dw = \mathbb{E}\big [T_k\big ] = {k\choose 2}^{-1}$.

\section{Discussion}
There is a natural curiosity in extending the stationary sampling distribution with small mutation rates (\ref{thm:c}), derived in \citet{BG2019}, to the sampling distribution taken at time $t$ using coalescent methods. An important point is that because of the duality of the coalescent process with the Wright-Fisher diffusion this is the same as the sampling distribution at time $t$ in a Wright-Fisher diffusion. 
In this paper we find expressions for the sampling distribution in a sample of $n$ genes taken at time $t$ from a Wright-Fisher population which follows a diffusion model. The overall mutation rate $\theta$ is taken to be small and mutation rates between types are general. An extension to gene trees is described. The idea in obtaining formulae is simple and is based on the result that to order $\theta$ it is only necessary to consider at most one mutation in sample lineages up to coalescence. A sample then consists of families of genes from founder lineages which may be of different types and at most one family of types descendent from a mutation in the coalescent tree. An illustration of these gene families is in Fig. \ref{fig:lineages}.
 There is a combinatorial complexity in Theorem 1 because of the different possible types at time zero. If the population at time zero consists of just one type, expressions for the sampling probability are much simpler as shown in Corollary 3.
 In a sample size of two the concepts are much easier, so that case is worked through first.  A simulation for sample size two shows the approximation to be very accurate for small $\theta$.

Many authors have observed that  samples at Single Nucleotide Polymorphism (SNP) sites with small mutation rates are either fixed, or contain bi-allelic states. If a site is fixed initially then explicit calculations in Corollary \ref{corollary:x} and Corollary \ref{corollary:4} model this feature of bi-allelic states.
An application of the theory in this paper is to SNP positions in DNA sequences where mutation rates are small.  Because the rate matrix $\gamma$ is constrained only 
by the requirement that its off-diagonal elements be ${\mathcal O}(\theta)$, the theory can accommodate any model class of nucleotide substitution rate matrices 
\citep[see, for instance][Table~1.1]{Yang06}.  

The frequency spectrum in DNA sequence data is often of interest. As an example we consider a frequency spectrum in which  
the ancestral bases are fixed at each individual SNP site, with relative proportions $z_A,z_T,z_C,z_G$ along the sequence sites. If $t << 1/\theta$ then Corollary \ref{corollary:4} applies at each SNP site and Eq. (\ref{config:00a}) shows that the probability of seeing $n_A$ entries $A$ at a site and $n_T$ entries $T$ when the ancestor type is $A$ is, to $\mathcal{O}(\theta^2)$,
\[
\gamma_{AT}\sum_{l=1}^n\sum_{k\geq l,2}^nkq_{nl}(k;t)p_{nk}(n_T),
\]
with similar expressions for other pairs. The expected proportion of sites which are SNPs for which there are $n_A$ entries $A$ and $n_T$ entries $T$ at the site 
when the ancestor base is not known, denoted by $f(n_A,n_T)$, is then
\[
f(n_A,n_T)=\sum_{l=1}^n\sum_{k\geq l,2}^nkq_{nl}(k;t)\Bigg (z_A\gamma_{AT}p_{nk}(n_T)+z_T\gamma_{TA}p_{nk}(n_A)\Bigg ).
\]
In a Bayesian approach a prior distribution could be chosen for $z_A,z_T,z_C,z_G$.
The derivation of $f(n_A,n_T)$ does not need independence between sites, which may be needed for further inference. If independence of SNP sites is assumed, then the likelihood of observing SNP site counts of $r(n_A,n_T)$ with $n_A$ entries $A$ and $n_T$ entries $T$ is $f(n_A,n_T)^{r(n_A,n_T)}$. Using a similar notation for the other base pairs the likelihood from segregating sites in the full data is, up to a multinomial coefficient, 
\begin{eqnarray*}
&&f(n_A,n_T)^{r(n_A,n_T)}
f(n_A,n_C)^{r(n_A,n_C)}
f(n_A,n_G)^{r(n_A,n_G)}
\nonumber \\
&&\times f(n_T,n_C)^{r(n_T,n_C)}
f(n_T,n_G)^{r(n_T,n_G)}
f(n_C,n_G)^{r(n_C,n_G)}.
\end{eqnarray*}
There is information in the non-segregating sites as well,  which could be included in the likelihood. This is a sketch of a likelihood application. Full details are beyond the scope of this paper.  Rate matrix estimation from site frequency data assuming stationarity is described in some detail in \cite{BT2017}.

There are several reasons to consider large sample sizes by taking $n\to \infty$. The first is that current sample sizes can be large. The second is that the approximate sampling distributions are obtained in probability to $o(\theta)$ for fixed sample size $n$, and the error can possibly be large when $n \to \infty$.
We show from (\ref{maxbound:0}) that when $n \to \infty$ this order of approximation can still be small if $n \to \infty$ and $\theta \to 0$ such that $n\theta \to \alpha$,  where $\alpha << 1$. The third reason is that ideally we would like to use the coalescent approach to find an approximate transition density in the Wright-Fisher diffusion for small mutation rates. The approach taken in a sample cannot be used directly in an infinite-leaf coalescent tree because as $\theta \to 0$ the probability of greater than one mutation in the coalescent tree is 1, not $\mathcal{O}(\theta^2)$. Nevertheless we can think of a large sample size model as approximating a discrete Wright-Fisher model with an effective population size of $N$. Then it is appropriate to consider sampling formulae when $n\to \infty$ and $\theta \to 0$ such that $n\theta \to \alpha << 1$  by thinking of $N=n$.
 $\alpha << 1$ is likely to be satisfied in applications, thinking of this as a population limit. Mutation rates at neutral genomic sites are typically less than  $10^{-7}$ per generation per base, effective population sizes $N \sim \mathcal{O}(10^4)$, so $\theta$ is typically $\mathcal{O}(10^{-3})$, and then $\alpha$ is typically less than $0.01$, which gives  the asymptotic probability of more than one mutation in an infinite coalescent tree limit, Eq.~(\ref{Q1forAlpha}), as 
 $Q_1 \sim  1 - e^{-\alpha}(1 + \alpha) < 0.00005$.

\section{Acknowledgement} We thank two referees for a careful reading of the manuscript and their comments and suggestions which have improved the paper.

\end{document}